\documentclass[11pt]{article}

\usepackage[T1]{fontenc}
\usepackage[utf8]{inputenc}
\usepackage{amsmath,amssymb,amsthm, amsfonts}
\usepackage{aliascnt}
\usepackage{mathtools}
\usepackage{algorithm}
\usepackage{algpseudocode}
\usepackage{booktabs}
\usepackage{graphicx}
\usepackage[margin=1in]{geometry}
\usepackage[colorlinks=true,allcolors=blue]{hyperref}
\usepackage[capitalise,noabbrev]{cleveref}

\algtext*{EndIf}
\algtext*{EndWhile}
\algtext*{EndFor}

\newtheorem{theorem}{Theorem}[section]
\newaliascnt{lemma}{theorem}
\newtheorem{lemma}[lemma]{Lemma}
\aliascntresetthe{lemma}
\newaliascnt{proposition}{theorem}
\newtheorem{proposition}[proposition]{Proposition}
\aliascntresetthe{proposition}
\newaliascnt{corollary}{theorem}
\newtheorem{corollary}[corollary]{Corollary}
\aliascntresetthe{corollary}
\theoremstyle{definition}
\newaliascnt{definition}{theorem}
\newtheorem{definition}[definition]{Definition}
\aliascntresetthe{definition}
\newaliascnt{assumption}{theorem}
\newtheorem{assumption}[assumption]{Assumption}
\aliascntresetthe{assumption}
\newaliascnt{remark}{theorem}
\newtheorem{remark}[remark]{Remark}
\aliascntresetthe{remark}
\newaliascnt{example}{theorem}
\newtheorem{example}[example]{Example}
\aliascntresetthe{example}

\crefname{equation}{equation}{equations}
\Crefname{equation}{Equation}{Equations}
\crefname{theorem}{theorem}{theorems}
\Crefname{theorem}{Theorem}{Theorems}
\crefname{lemma}{lemma}{lemmas}
\Crefname{lemma}{Lemma}{Lemmas}
\crefname{proposition}{proposition}{propositions}
\Crefname{proposition}{Proposition}{Propositions}
\crefname{corollary}{corollary}{corollaries}
\Crefname{corollary}{Corollary}{Corollaries}
\crefname{definition}{definition}{definitions}
\Crefname{definition}{Definition}{Definitions}
\crefname{assumption}{assumption}{assumptions}
\Crefname{assumption}{Assumption}{Assumptions}
\crefname{remark}{remark}{remarks}
\Crefname{remark}{Remark}{Remarks}
\crefname{example}{example}{examples}
\Crefname{example}{Example}{Examples}
\crefname{algorithm}{algorithm}{algorithms}
\Crefname{algorithm}{Algorithm}{Algorithms}
\crefname{section}{section}{sections}
\Crefname{section}{Section}{Sections}

\DeclareMathOperator{\depth}{depth}
\DeclareMathOperator{\diam}{diam}
\newcommand{\E}{\mathbb{E}}

\newcommand{\priv}{\mathrm{priv}}
\newcommand{\eps}{\varepsilon}

\makeatletter
\providecommand{\theHALG@line}{}
\renewcommand{\theHALG@line}{\thealgorithm.\arabic{ALG@line}}
\makeatother

\title{Intrinsic-Dimensional Wasserstein Guarantees for Private Synthetic Measures}
\author{Yiyun He\thanks{Department of Mathematics, University of California, San Diego. Email: \texttt{yih130@ucsd.edu}}}
\date{September 15, 2026}

\begin{document}

\maketitle

\begin{abstract}
We study an $\eps$-differentially private synthetic measure for $n$ points in $[0,1]^d$ by applying the existing PrivTree algorithm \cite{zhang2016privtree} to construct an adaptive binary partition and then privately releasing its leaf masses.
We consider the worst-case data model without any sampling or population-distribution assumption.
The 1-Wasserstein error of the synthetic measure is $\widetilde O_d((\eps n)^{-1/d})$ for $d\ge2$, which is optimal compared to the minimax lower bound \cite{boedihardjo2024private} up to a logarithmic factor.

Moreover, for $d\ge3$ and $2<s\le d$, if the data set has covering number at most $Ar^{-s}$ over the relevant finite range of scales $r$, the expected error improves to $\widetilde O_{d,s}((\eps n)^{-1/s})$.
Thus the rate depends on a finite-scale intrinsic dimension rather than the ambient dimension, without requiring the recovery of a low-dimensional manifold.
We also introduce a shifting technique to further avoid the exponential dependence of the constant on the ambient dimension $d$.
\end{abstract}


\section{Introduction}

Differential privacy, as one of the strongest privacy concepts, aims to protect each individual's information by introducing randomness in the algorithms \cite{dwork2006calibrating,dwork2014algorithmic}.
It requires similar output distributions when one record is added or removed in the input, so an observer cannot distinguish whether a specific individual's data is included in the real data set.
The 2020 U.S. Census applied differential privacy in a public data release at national scale.
This deployment highlights both the practical importance of privacy protection and the challenge of preserving accuracy across the many uses of released census data \cite{hawes2020implementing,abowd2022topdown,hauer2021census}.

Much of the differential privacy literature focuses on specific tasks, including private mean estimation \cite{liu2021robustmean} and covariance estimation \cite{amin2019covariance}.
Examples in machine learning include linear regression \cite{sheffet2017ols}, clustering \cite{ghazi2020clustering}, and stochastic-gradient training \cite{abadi2016deep}, etc.
These algorithms add privacy guarantees to the chosen output, and multiple uses of the real data usually require a larger privacy budget.
In contrast, a private synthetic data set can be released once and reused for many analyses without further access to the sensitive data \cite{blum2013learning,hardt2012simple,boedihardjo2024covariance,kreacic2023kdtrees}.
The postprocessing property of differential privacy incurs no additional privacy loss despite any further use of the private synthetic data.
For example, the Private Evolution algorithm is a recent training-free example that uses public foundation-model APIs to generate private synthetic data \cite{lin2024foundation}.
Its recent analysis gives a worst-case Wasserstein guarantee under approximate differential privacy and a Gaussian variation oracle \cite{gonzalezlara2025private}.
 
For synthetic measures, we measure the utility with the Wasserstein-1 distance between the empirical measures of the input and output.
Also known as the earth mover's distance, the Wasserstein distance is the minimum cost of transporting one probability measure to another.
Kantorovich--Rubinstein duality \cite{villani2009optimal} shows that the Wasserstein distance is the supremum of the integration error over all $1$-Lipschitz functions. Thus, a small Wasserstein error controls the difference in expectations of every Lipschitz real-valued statistic, including Lipschitz scores used in distance-based classification \cite{vonluxburg2004distance}. (See \Cref{sec:wasserstein}.)
On $[0,1]^d$, it has been shown that the upper bound of worst-case error as well as corresponding minimax lower bound match at order $n^{-1/d}$ \cite{boedihardjo2024private,he2023algorithmically} up to a constant.
The exponent $1/d$ reflects the \emph{curse of high dimensionality}: $d\lesssim \log n$ is necessary to have a nontrivial accuracy, and the rate is still practically worthless when $d$ is even a large constant due to the slow convergence rate.
However, this barrier is inevitable in general: Ullman and Vadhan show that, assuming the existence of the one-way functions, no polynomial-time private algorithm can output a synthetic database on the Boolean cube that accurately preserves all two-way marginals \cite{ullman2020pcps}. 

To tackle the curse of dimensionality, one solution is to restrict the query class, as preserving all Lipschitz queries might be too ambitious for many applications.
Some distance weaker than the Wasserstein distance may incur better rates.
For example, only considering $s$-sparse Lipschitz queries admits an error of $\widetilde O((\eps n)^{-1/s})$ \cite{donhauser2024certified}, while $k$-smooth queries admit an error matching the minimax exponent $\min\{1,k/d\}$ up to logarithmic factors with approximate differential privacy \cite{ding2026smooth}.
Another approach to reduce the dependence on the ambient dimension takes advantage of low-dimensional geometry of the input data.
We aim to replace the ambient dimension $d$ in the exponent by the intrinsic dimension of the data \cite{boedihardjo2024private, donhauser2024certified,he2025lowdimensional}.
Related non-private work also gives ambient-dimension-dependent empirical Wasserstein rates \cite{fournier2015rate} and refinements that replace the ambient dimension by Wasserstein dimensions of the population distribution \cite{weed2019sharp}, which encourage us to pursue the private version of polynomial algorithms with similar error rates.

To exploit low-dimensional geometry without first estimating a subspace or manifold, we use a multiscale partition of the ambient space, proposed in \cite{zhang2016privtree}.
Hierarchical partitions are a classical multiscale tool in optimal transport: they match mass locally inside fine cells and move only the remaining mass across coarser cells \cite{ba2011sublinear,dereich2013constructive,weed2019sharp}.
This idea also underlies the private construction in \cite{he2023algorithmically}, which adds noise to counts on a uniform binary partition.
PrivTree instead makes the hierarchy data dependent: it recursively splits a cell when its biased noisy count exceeds a threshold, while the noise scale is independent of depth \cite{zhang2016privtree}.

\subsection{Main results}

Our construction has two stages.
First, we apply PrivTree \cite{zhang2016privtree} to a cyclic binary partition of $[0,1]^d$.
Cells with large biased noisy masses are refined, while lighter cells remain coarse, so the tree adapts to the geometry of the data.
After releasing the tree, we add independent Laplace noise to its leaf masses and then reconstruct a probability measure.
This separates the Wasserstein error into a resolution term from aggregating each leaf to one representative point and a privacy term from perturbing the leaf masses.
We analyze these two terms with the following tools.

Our first tool is a deterministic resolution functional: it sums the mass of each terminal cell times the scale of that cell.
It bounds the cost of aggregating the leaves to representative points, especially when the tree is unbalanced.
The second tool is a data-dependent Laplacian complexity bound for the privacy error, which adapts to the geometry of the data.
For a data set of size $m$ and intrinsic dimension $s>2$, it has order $m^{-1/s}$, replacing the ambient dimension $d$ by $s$ and overcoming the curse of dimensionality.

Combining these two tools with the privacy analysis gives the following informal result.

\begin{theorem}[Informal]\label{thm:informal-main}
Assume that $0<\eps\le1$ and $\eps n>2$.
The resulting synthetic measure is $\eps$-differentially private, has expected support size $O(\eps n)$, and satisfies
\[
    \mathbb E W_1(\mu_n,\widehat\nu)
    \le C\sqrt d
    \begin{cases}
        (\eps n)^{-1/2}, & d=1,\\
        (\eps n)^{-1/2}\log(\eps n), & d=2,\\
        \left(\dfrac{\log(\eps n)}{\eps n}\right)^{1/d}, & d\ge3.
    \end{cases}
\]
Moreover, let $d\ge3$, $2<s\le d$, $1\le A<\eps n$, and $X=\{X_i\}_{i=1}^n\subseteq\Omega$.  Suppose the covering number bound
\[
    N(X,r)\le Ar^{-s}
    \qquad\text{for every } \left(\frac A{\eps n}\right)^{1/s}\le r\le1.
\]
Then
\[
    \mathbb E W_1(\mu_n,\widehat\nu)
    \le C_s d^{3/2}\left(\frac{Ad\log(\eps n)}{\eps n}\right)^{1/s}.
\]
\end{theorem}

\paragraph{Discussion.}
\emph{Resolution scale.} We do not need the covering bound at arbitrarily small $r$. For any candidate $s$, it is enough to verify the bound down to $r=(A/(\eps n))^{1/s}$, because every finite data set has Minkowski dimension $0$ as $r\downarrow0$. A natural cutoff $r^*$ (before determining $s$) is the scale at which the covering number reaches the effective private sample size, namely $N(X,r^*)\asymp\eps n$; under the power-law bound $N(X,r)\le Ar^{-s}$, this balance gives $r^*\asymp(A/(\eps n))^{1/s}$.

\emph{Minimax lower bound.} The result with exponent $1/s$ matches the minimax lower bound up to logarithmic factors for fixed privacy parameter $\eps$. Indeed, consider the case where $s$ is an integer and only the first $s$ coordinates vary, that is, the data lies in $[0,1]^s\times\{0\}^{d-s}$. In any possible definition of intrinsic dimension, as long as this particular case is regarded as $s$-dimensional, the minimax lower bound on $[0,1]^s$ in \cite{boedihardjo2024private} gives error $\Omega(n^{-1/s})$.

\emph{Time complexity.} The adaptive partition \Cref{alg:binary-privtree} has expected tree size $O(\varepsilon n)$, and hence running time $\mathrm{poly}(n,d)$ for fixed $\eps$. In the synthetic data generation, \Cref{alg:private-synthetic-measure}, which has time complexity $\mathrm{poly}(n,d)$ as studied in \cite{he2023algorithmically}, dominates the time complexity.

\paragraph{Comparison with existing results.}
For arbitrary data and $d\ge2$, \Cref{thm:informal-main} has the exponent $1/d$, matching the minimax lower bound \cite{boedihardjo2024private}, yet the best algorithm using the uniform binary construction in \cite{he2023algorithmically} already attains $O_d((\eps n)^{-1/d})$ without our logarithmic factor and runs in linear time.
The main contribution of our result is to use adaptive partitions to avoid the curse of high-dimensionality by capturing the low-dimensional structure of data.
The construction in \cite{donhauser2024certified} also adapts to data on an $s$-dimensional covering manifold and more general unknown nonlinear structure, but its exponent is $1/(s+1)$, and its implementation has exponential complexity.
For data lying on an unknown $s$-dimensional affine subspace, the polynomial-time method in \cite{he2025lowdimensional} attains exponent $1/s$ without a spectral-gap assumption, which can be recovered by our result with better constant dependence in $d$.
The instance-optimal framework in \cite{feldman2024instance} embeds finite metric spaces into hierarchically separated trees, where the Wasserstein distance has a simple tree representation. Its utility guarantee requires i.i.d. sampling from the population distribution, while ours focuses on the worst-case scenario and holds conditionally for every fixed data set.
Related hierarchical methods give sparse EMD guarantees for two-dimensional heatmaps \cite{ghazi2023heatmaps}; they use a fixed grid and an input model in which each user contributes a distribution.

\paragraph{Organization.}
The paper is organized as follows.
\Cref{sec:preliminaries} introduces important concepts and notation used throughout the paper.
\Cref{sec:aggregation-laplacian} describes the private leaf-mass release and proves the data-dependent Laplacian complexity bound.
\Cref{sec:warmup}, as a warm-up for the adaptive partitions, develops the deterministic hard-threshold partitions and provides tools for later analysis.
\Cref{sec:soft-privtree} analyzes soft thresholds and PrivTree, getting a worst-case error bound for accuracy and also intrinsic-dimensional guarantees.
\Cref{sec:random-shift} replaces the fixed grid by a public random shift and states the resulting guarantee in terms of covering numbers.

\section{Preliminaries}\label{sec:preliminaries}

\subsection{Wasserstein distance and empirical measures}
\label{sec:wasserstein}

Throughout, let $d\ge1$ and let $\Omega=[0,1]^d$, equipped with the Euclidean metric.
For fixed points $X_1,\ldots,X_n\in\Omega$, their empirical measure is
\begin{equation}
    \mu_n\coloneqq\frac1n\sum_{i=1}^n\delta_{X_i}.
\end{equation}
For probability measures $\mu$ and $\nu$ on $\Omega$, write $\Pi(\mu,\nu)$ for the set of their couplings and define
\[
    W_1(\mu,\nu)
    \coloneqq\inf_{\pi\in\Pi(\mu,\nu)}
      \int_{\Omega\times\Omega}\|x-y\|_2\,d\pi(x,y).
\]
The quantity $W_1(\mu,\nu)$ is also called the earth-mover distance: it is the minimum cost of rearranging mass distributed according to $\mu$ into mass distributed according to $\nu$, where moving an amount $m$ over a distance $\ell$ costs $m\ell$.
Equivalently, a coupling $\pi$ specifies how much mass is transported from each location $x$ to each location $y$.

Let
\[
    \mathcal F
    \coloneqq\{f:\Omega\to\mathbb R:
       \operatorname{Lip}(f)\le1,\ \|f\|_\infty\le\sqrt d\}
\]
be the bounded $1$-Lipschitz function class on $\Omega$.  The Kantorovich--Rubinstein duality \cite{villani2009optimal} gives
\begin{equation}\label{eq:kr-duality}
    W_1(\mu,\nu)
    =\sup_{\operatorname{Lip}(f)\le1}
       \left|\int_\Omega f\,d\mu-\int_\Omega f\,d\nu\right|
    =\sup_{f\in\mathcal F}
       \left|\int_\Omega f\,d\mu-\int_\Omega f\,d\nu\right|.
\end{equation}
Indeed, adding a constant to $f$ does not change either difference because $\mu$ and $\nu$ are probability measures.  We may therefore impose $f(0)=0$, which gives $\|f\|_\infty\le\diam(\Omega)=\sqrt d$.

\subsection{Differential privacy}

We next define privacy for the sample $X_1,\ldots,X_n$.
Differential privacy requires the output distribution to change little when one observation is added or removed.

\begin{definition}[Differential privacy, \cite{dwork2014algorithmic}]
  Two data sets $D,D'$ are adjacent if they differ by one data point, that is, one can be obtained from the other by adding or removing one data point.
  A randomized mechanism $\mathcal M$ is $\eps$-differentially private if, for every adjacent data sets $D, D'$ and every measurable set $S$ of outputs,
  \[
    \mathbb P\{\mathcal M(D)\in S\}
    \le e^{\eps}\cdot \mathbb P\{\mathcal M(D')\in S\}.
  \]
\end{definition}

Our final mechanism first releases a private tree and then private leaf masses based on the tree. Therefore, we use the following principle of adaptive composition of differentially private mechanisms.

\begin{lemma}[Adaptive composition, \cite{dwork2014algorithmic} Theorem~B.1]\label{lem:adaptive-composition}
Let $\mathcal M_1$ be $\eps_1$-differentially private, and suppose that $D\mapsto\mathcal M_2(D,y)$ is $\eps_2$-differentially private for every fixed output $y$ of $\mathcal M_1$.  Then
\[
    D\longmapsto\bigl(\mathcal M_1(D),\mathcal M_2(D,\mathcal M_1(D))\bigr)
\]
is $(\eps_1+\eps_2)$-differentially private.
\end{lemma}

To guarantee privacy in the leaf-mass perturbation step, we use the Laplace mechanism, one of the simplest and most fundamental mechanisms in differential privacy.

\begin{lemma}[Laplace mechanism and post-processing, \cite{dwork2014algorithmic} Theorem~3.6]\label{lem:laplace-mechanism}
For a query $q$ taking values in $\mathbb R^m$, define its $\ell_1$-sensitivity by
\[
    \Delta_1(q)\coloneqq\sup_{D\sim D'}\|q(D)-q(D')\|_1,
\]
where the supremum is over adjacent data sets.
Here $\operatorname{Lap}(b)$ denotes the centered Laplace distribution with scale $b$.
If $\lambda_1,\ldots,\lambda_m$ are independent $\operatorname{Lap}(\Delta_1(q)/\eps)$ variables, then $q(D)+(\lambda_1,\ldots,\lambda_m)$ is $\eps$-differentially private.
\end{lemma}

\subsection{Binary partition tree}

We use a geometrically balanced cyclic binary partition.
Let $\mathcal D_0=\{\Omega\}$.
To obtain $\mathcal D_{k+1}$, bisect every cell of $\mathcal D_k$ into two congruent rectangles, orthogonally to coordinate $1+(k\bmod d)$.
Thus each coordinate is bisected once in every block of $d$ levels.
We fix a consistent convention for assigning points on cell boundaries and denote by $Q_k(x)\in\mathcal D_k$ the cell containing $x$ at depth $k$.
A depth-$k$ cell $Q$ satisfies
\begin{equation}\label{eq:diameter-depth}
    \diam(Q)
    \le \sqrt d\,2^{-\lfloor k/d\rfloor}
    \le 2\sqrt d\,2^{-k/d}.
\end{equation}

A partition tree $\mathcal T$ is a finite rooted subtree of this complete binary partition whose leaves form a partition of $\Omega$.
Write $\mathcal L(\mathcal T)$ for its leaves and $\depth(Q)$ for the depth of a cell $Q$ as a node in the tree.
For any $x\in\Omega$, let $K_{\mathcal T}(x)$ be the depth of the leaf containing $x$.
For simplicity, set
\[
    \alpha\coloneqq2^{-1/d} \in [1/2, 1),\qquad
    p_Q\coloneqq\mu(Q).
\]
Given representatives $y_Q\in Q$ for the leaves of $\mathcal T$, define the exact synthetic measure associated with $\mu$ and $\mathcal T$, together with its corresponding resolution functional, by
\begin{align}
    \nu_{\mathcal T}
    &\coloneqq\sum_{Q\in\mathcal L(\mathcal T)}p_Q\delta_{y_Q},
    \\
    \mathfrak R_{\mathcal T}(\mu)
    &\coloneqq \E_{x\sim \mu}[\alpha^{K_{\mathcal T}(x)}]
    = \int_\Omega 2^{-K_{\mathcal T}(x)/d}\,d\mu(x)
      =\sum_{Q\in\mathcal L(\mathcal T)}p_Q2^{-\depth(Q)/d}.
\end{align}
The factor $\alpha^{K_{\mathcal T}(x)}$ records the geometric scale of the leaf containing $x$, so $\mathfrak R_{\mathcal T}(\mu)$ is the average resolution under $\mu$.  It directly controls the cost of aggregating each leaf to its representative.

The next lemma, which is the most commonly used in this paper, bounds the resolution error of the partition tree in terms of the resolution functional $\mathfrak R_{\mathcal T}(\mu)$. Therefore, to study the resolution error of the aggregation step, it suffices to study the resolution functional $\mathfrak R_{\mathcal T}(\mu)$.

\begin{lemma}[Resolution error]
\label{lem:resolution-error}
For every partition tree $\mathcal T$ and every choice of leaf representatives,
\begin{equation}
    W_1(\mu,\nu_{\mathcal T})
    \le \sum_{Q\in\mathcal L(\mathcal T)}p_Q\diam(Q)
    \le 2\sqrt d\,\mathfrak R_{\mathcal T}(\mu).
\end{equation}
\end{lemma}

\begin{proof}
Couple each $x\in Q$ with the representative $y_Q$.  The cost on leaf $Q$ is at most $p_Q\diam(Q)$, which gives the first inequality after summation.  The second follows from \Cref{eq:diameter-depth}.
\end{proof}

Since every depth-$k$ cell has volume $2^{-k}$ and the leaves partition $\Omega$, the volume argument (or Kraft's identity) gives
\begin{equation}\label{eq:kraft}
    \sum_{Q\in\mathcal L(\mathcal T)}2^{-\depth(Q)}=1.
\end{equation}
The following result is useful in subsequent computations.

\begin{lemma}
\label{lem:leafwise-holder}
The resolution functional satisfies
\begin{equation}\label{eq:leafwise-holder}
    \mathfrak R_{\mathcal T}(\mu)
    =\sum_{Q\in\mathcal L(\mathcal T)}p_Q2^{-\depth(Q)/d}
    \le\left(\sum_{Q\in\mathcal L(\mathcal T)}p_Q2^{-\depth(Q)}\right)^{1/d}.
\end{equation}
\end{lemma}

\begin{proof}
For $d=1$, the assertion is an equality. For $d>1$, H\"older's inequality gives
\[
\begin{split}
    \sum_{Q\in\mathcal L(\mathcal T)}p_Q2^{-\depth(Q)/d}
    &=\sum_{Q\in\mathcal L(\mathcal T)}p_Q^{1-1/d}
       \bigl(p_Q2^{-\depth(Q)}\bigr)^{1/d} \\
    &\le\left(\sum_{Q\in\mathcal L(\mathcal T)}p_Q\right)^{1-1/d}
       \left(\sum_{Q\in\mathcal L(\mathcal T)}p_Q2^{-\depth(Q)}\right)^{1/d} \\
    &=\left(\sum_{Q\in\mathcal L(\mathcal T)}p_Q2^{-\depth(Q)}\right)^{1/d}.
\end{split}
\]
\end{proof}

\subsection{Intrinsic dimensions}

The Wasserstein error for full-dimensional data has rate of order $n^{-1/d}$, which becomes slow when $d$ is large \cite{boedihardjo2024private}.  We therefore use the number of occupied cells to describe lower-dimensional structure.  At depth $k$, define
\[
    N_k^+(\mu)\coloneqq\#\{Q\in\mathcal D_k:\mu(Q)>0\}.
\]
This quantity records how many cells at geometric scale about $2^{-k/d}$ carry positive mass.  It can grow as fast as $2^k$ for a full-dimensional measure such as $\mathrm{Unif}([0,1]^d)$, but more slowly when the measure is concentrated near a lower-dimensional set.  We use the following assumption throughout the intrinsic-dimensional results.

\begin{assumption}[Intrinsic dimension]\label{ass:finite-scale-intrinsic-dimension}
Fix $A\ge1$, $0<s\le d$, and a cutoff parameter $R>A$, and set
\[
    K_s\coloneqq\left\lceil\frac d s\log_2\!\left(\frac RA\right)\right\rceil.
\]
We say that $\mu$ satisfies the finite-scale intrinsic-dimension assumption with parameters $(A,s,R)$ if
\begin{equation}\label{eq:finite-scale-intrinsic-dimension}
    N_k^+(\mu)\le A2^{sk/d}
    \qquad\forall\,0\le k\le K_s.
\end{equation}
\end{assumption}

In the assumption, the parameter $s$ describes the intrinsic dimension of $\mu$, while $A$ allows for a fixed geometric constant. The parameter $R$ determines how far down the tree the growth condition is required to hold: $K_s$ is the first depth at which $A2^{sk/d}$ reaches $R$, and no restriction is imposed beyond this depth.

\begin{remark}[Relation to Minkowski dimension]
    For a bounded set $S\subset\mathbb R^d$, let $N(S,r)$ be the smallest number of Euclidean balls of radius $r$ needed to cover $S$.  Its upper Minkowski dimension is
    \[
    \overline{\dim}_{\mathrm M}(S)
    \coloneqq\limsup_{r\downarrow0}
       \frac{\log N(S,r)}{\log(1/r)}.
    \]
\begin{enumerate}
    \item
    At depth $k$, our cells have scale about $2^{-k/d}$, but converting a covering number $N(S,2^{-k/d})$ to fixed-grid occupancy $N^+_k$ can lose a factor exponential in $d$: a small ball aligned with grid boundaries may meet $2^d$ cells, and incomplete blocks of coordinate splits cause a similar loss.  \Cref{sec:random-shift} replaces the fixed grid by a public random shift and reduces this covering-to-occupancy loss to a polynomial factor for fixed intrinsic dimension.  
    
    \item Unlike Minkowski dimension, which takes $r\downarrow0$, \Cref{ass:finite-scale-intrinsic-dimension} only requires this growth down to $r=2^{-K_s/d}$.

    \item If the data indeed lies on a compact $s$-dimensional smooth manifold $M\subset[0,1]^d$, then $M$ can be covered by at most $C_{M}r^{-s}$ radius-$r$ balls.  Comparing this covering with the grid gives
    \[
        N_k^+(\mu_n)\le \min\{n,A2^{sk/d}\}
    \]
    for a constant $A$ depending only on $M$.  Therefore the data satisfy \Cref{ass:finite-scale-intrinsic-dimension} over the required range of depths.  In the special case $M=[0,1]^s\times\{0\}^{d-s}$, one may take $A=2^s$ because $N_k^+(\mu_n)\le2^{s\lceil k/d\rceil}\le2^s2^{sk/d}$.
\end{enumerate}
\end{remark}

\section{Aggregation algorithms and Laplacian complexity}\label{sec:aggregation-laplacian}

We will consider various partitioning algorithms, whose outputs are the partition trees.
For a given tree, we record here how to transform it into a private synthetic measure.

\begin{algorithm}[H]
\caption{Simple aggregation}
\label{alg:exact-synthetic-measure}
\begin{algorithmic}[1]
\Require Probability measure $\mu$ and a partition tree $\mathcal T$
\ForAll{$Q\in\mathcal L(\mathcal T)$}
    \State Pick $y_Q\in Q$ arbitrarily and set $p_Q\gets\mu(Q)$.
\EndFor
\State \Return $\nu_{\mathcal T}\gets\sum_{Q\in\mathcal L(\mathcal T)}p_Q\delta_{y_Q}$.
\end{algorithmic}
\end{algorithm}

To guarantee the privacy, simply aggregating the leaf masses $p_Q$ is not sufficient.
For an empirical input $\mu_n$, the following algorithm applies the Laplace mechanism in \Cref{lem:laplace-mechanism} to the leaf masses.
The privacy of \Cref{alg:private-synthetic-measure} is ensured when the tree $\mathcal T$ is either public or released in a differentially private manner. 
In the latter case, we can assign $\eps/2$ privacy budget to the tree release and synthetic measure release each to obtain $\eps$-differential privacy.

\begin{algorithm}[H]
\caption{Private synthetic measure from a given tree}
\label{alg:private-synthetic-measure}
\begin{algorithmic}[1]
\Require Empirical measure $\mu_n$, released tree $\mathcal T$, and privacy budget $\eps>0$
\ForAll{$Q\in\mathcal L(\mathcal T)$}
    \State Choose $y_Q\in Q$ by a fixed data-independent rule.
\EndFor
\State Draw $\widetilde p_Q\sim\mu_n(Q)+\operatorname{Lap}(1/(\eps n))$ independently for every leaf $Q$.
\State Compute the closest probability measure $(\widehat p_Q)_{Q\in\mathcal L(\mathcal T)}$ by \Cref{eq:psmm-projection}.
\State \Return $\widehat\nu_{\mathcal T}\gets\sum_{Q\in\mathcal L(\mathcal T)}\widehat p_Q\delta_{y_Q}$.
\end{algorithmic}
\end{algorithm}

The projection in the last step, formally stated below, is to find the closest probability measure $(\widehat p_Q)_{Q\in\mathcal L(\mathcal T)}$ to the signed measure $(\widetilde p_Q)_{Q\in\mathcal L(\mathcal T)}$ under integral probability metrics with respect to $\mathcal F$. 
\begin{equation}\label{eq:psmm-projection}
    (\widehat p_Q)_{Q\in\mathcal L(\mathcal T)}
    \in\underset{\substack{p_Q\ge0\\ \sum_Q p_Q=1}}{\arg\min}
    \ \max_{\substack{|f_Q|\le\sqrt d\\
                       |f_Q-f_R|\le\|y_Q-y_R\|_2}}
       \sum_{Q\in\mathcal L(\mathcal T)}(\widetilde p_Q-p_Q)f_Q.
\end{equation}
This is the PSMM projection of \cite[Algorithm~2]{he2023algorithmically}; an explicit linear program is given in Appendix~\ref{app:psmm-projection}.

Under the structure of \Cref{alg:private-synthetic-measure}, the expected 1-Wasserstein error of the private synthetic measure has two components:
\[W_1\text{ error of synthetic data} \leq \text{resolution error} + \text{privacy noise error}.\]
Here the \emph{resolution error} comes from aggregating the leaf masses to a single representative in each leaf cell, and therefore depends on the tree structure. We will analyze it later in \Cref{sec:warmup,sec:soft-privtree}. 
The \emph{privacy noise error} comes from the Laplace noise $(\lambda_Q)_{Q\in \mathcal L(\mathcal T)}$, and we will analyze it via Laplacian complexity.

\subsection{Laplacian complexity}

For a data set $Z=(z_1,\ldots,z_m)\in\Omega^m$, define its Laplacian complexity and the corresponding worst-case Laplacian complexity by
\begin{equation}\label{eq:laplacian-complexity}
\begin{split}
    \mathfrak L(\mathcal F;Z)
    &\coloneqq
    \E\left[
       \sup_{f\in\mathcal F}
       \left|\frac1m\sum_{i=1}^m\lambda_i f(z_i)\right|
    \right],
    \\
    \mathfrak L_m(\mathcal F)
    &\coloneqq
    \sup_{|Z|=m}\mathfrak L(\mathcal F;Z),
\end{split}
\end{equation}
where $\lambda_1,\ldots,\lambda_m\sim\operatorname{Lap}(1)$ are independent, as in \cite[Definition~2]{he2023algorithmically}.

\begin{proposition}
\label{prop:private-leaf-mass-error}
Fix a given finite tree $\mathcal T$ with $m=|\mathcal L(\mathcal T)|$, and denote $Z_{\mathcal T}=(y_Q)_{Q\in\mathcal L(\mathcal T)}$ for the leaf representatives.
Then \Cref{alg:private-synthetic-measure} is $\eps$-differentially private and
\begin{equation}\label{eq:private-leaf-mass-error}
    \mathbb E\bigl[W_1(\nu_{\mathcal T},\widehat\nu_{\mathcal T})\bigr]
    \le \frac{2m}{\eps n}\mathfrak L(\mathcal F;Z_{\mathcal T})
    \le \frac{2m}{\eps n}\mathfrak L_m(\mathcal F)
    \le \frac{C\sqrt d}{\eps n}
    \begin{cases}
        m^{1/2}, & d=1,\\
        m^{1/2}\log(1+m), & d=2,\\
        m^{1-1/d}, & d\ge3.
    \end{cases}
\end{equation}
Here $C$ is universal, and the expectation is over the fresh leaf-mass noise.
\end{proposition}

We refer to \cite{he2023algorithmically} for the original proof.  To keep the presentation self-contained, we include the short argument below.

\begin{proof}
For fixed $\mathcal T$, the leaf-mass vector has $\ell_1$ sensitivity $1/n$, so \Cref{lem:laplace-mechanism} proves privacy after post-processing.  For accuracy, the triangle inequality and the optimality of the projection in \Cref{eq:psmm-projection} give
\begin{align*}
    W_1(\nu_{\mathcal T},\widehat\nu_{\mathcal T})
    &\le \sup_{f\in\mathcal F}\left|\sum_Q(p_Q-\widetilde p_Q)f(y_Q)\right|
       +\sup_{f\in\mathcal F}\left|\sum_Q(\widetilde p_Q-\widehat p_Q)f(y_Q)\right| \\
    &\le 2\sup_{f\in\mathcal F}\left|\sum_Q(\widetilde p_Q-p_Q)f(y_Q)\right|\\
    & = \frac{2}{\eps n}\sup_{f\in\mathcal F}\left|\sum_Q\lambda_Q f(y_Q)\right|.
\end{align*}
Since $\lambda_Q$ are independent $\operatorname{Lap}(1)$ variables, taking expectations proves the first inequality in \Cref{eq:private-leaf-mass-error}.  
The second follows from the definition of $\mathfrak L_m(\mathcal F)$, and the last follows from \cite[Corollary~4]{he2023algorithmically}.  Since $\|x-y\|_2\le\sqrt d\,\|x-y\|_\infty$, passing from the $\ell_\infty$ metric to the Euclidean metric costs at most the factor $\sqrt d$ displayed above; the remaining constant $C$ is universal.
\end{proof}

\subsection{Dimension dependence in Laplacian complexity}

The last bound in \Cref{eq:private-leaf-mass-error} uses the worst-case complexity $\mathfrak L_m(\mathcal F)$ and therefore does not use the geometry of the particular data set $Z$.  This is appropriate when the points fill the ambient space, but it can be loose when they are concentrated near a lower-dimensional set.
We now give a data-dependent bound that can exploit lower-dimensional structure of $Z$.

For $Z=(z_1,\ldots,z_m)\in\Omega^m$, let
\[
    \mu_Z\coloneqq\frac1m\sum_{i=1}^m\delta_{z_i}
\]
and define its effective occupancy at depth $k$ by
\begin{equation}\label{eq:effective-occupancy}
    N_k^{\mathrm{eff}}(Z)
    \coloneqq
    \left(\sum_{Q\in\mathcal D_k}
       \sqrt{\mu_Z(Q)}\right)^2.
\end{equation}
The identity $\sum_{Q\in \mathcal D_k} \mu_Z(Q)=1$ and the Cauchy--Schwarz inequality give
\begin{equation}\label{eq:effective-occupancy-comparison}
    1\le N_k^{\mathrm{eff}}(Z)\le N_k^+(\mu_Z).
\end{equation}
Indeed, we have
\[
    1
    =\left(\sum_{Q\in\mathcal D_k}\mu_Z(Q)\right)^2
    \le\left(\sum_{Q\in\mathcal D_k}\sqrt{\mu_Z(Q)}\right)^2
    =N_k^{\mathrm{eff}}(Z)
    \le N_k^+(\mu_Z)\left(\sum_{Q\in\mathcal D_k}\mu_Z(Q)\right)
    =N_k^+(\mu_Z).
\]

The next proposition gives us a data-dependent bound on the Laplacian complexity that allows us to improve the $O(m^{-1/d})$ dependence via capturing the lower-dimensional structure of $Z$.
\begin{proposition}
\label{prop:data-dependent-laplacian-complexity}
For every $Z\in\Omega^m$ and every integer $J\ge0$,
\begin{equation}\label{eq:laplacian-multiscale-bound}
\begin{split}
    \mathfrak L(\mathcal F;Z)
    \le{}& \sqrt{\frac{2d}{m}}
       +\frac{\sqrt{2d}}{\sqrt m}
         \sum_{j=1}^J2^{-(j-1)}\sqrt{N_{dj}^{\mathrm{eff}}(Z)}
       +\sqrt d\,2^{-J}.
\end{split}
\end{equation}
\end{proposition}

\begin{proof}
We will run a classical chaining argument to compute the expectation of the supremum.
For every occupied cell $Q\in\mathcal D_{dj}$, $0\le j\le J$, choose $y_Q\in Q\cap Z$.  For $z_i\in Z$, let $y_j(z_i)$ be the representative of its depth-$dj$ cell.  Every $f\in\mathcal F$ satisfies
\[
\begin{split}
    f(z_i)
    =f(y_0(z_i))
     +\sum_{j=1}^J\bigl(f(y_j(z_i))-f(y_{j-1}(z_i))\bigr)
     +f(z_i)-f(y_J(z_i)).
\end{split}
\]
The two representatives in the $j$th increment lie in the same depth-$d(j-1)$ cube, while $z_i$ and $y_J(z_i)$ lie in the same depth-$dJ$ cube.  Hence
\[
    |f(y_j(z_i))-f(y_{j-1}(z_i))|
    \le\sqrt d\,2^{-(j-1)},
    \qquad
    |f(z_i)-f(y_J(z_i))|
    \le\sqrt d\,2^{-J}.
\]
For every occupied cell $Q$, set $\Lambda_Q\coloneqq\sum_{i:z_i\in Q}\lambda_i$.  Since the Laplace variables have variance $2$,
\[
    \E|\Lambda_Q|
    \le\sqrt{\E[\Lambda_Q^2]}
    =\sqrt{2m\mu_Z(Q)}.
\]

Let $Q_0=\Omega$, and for $Q\in\mathcal D_{dj}$ with $j\ge1$, let $Q^-$ denote its depth-$d(j-1)$ ancestor.  Grouping the $j$th increment according to the cell containing $z_i$ gives
\[
    \sum_{i=1}^m\lambda_i
       \bigl(f(y_j(z_i))-f(y_{j-1}(z_i))\bigr)
    =\sum_{\substack{Q\in\mathcal D_{dj}\\\mu_Z(Q)>0}}
       \Lambda_Q\bigl(f(y_Q)-f(y_{Q^-})\bigr).
\]
Applying the triangle inequality across the root term, the $J$ increments, and the remainder yields
\begin{align*}
 \sup_{f\in\mathcal F}
   \left|\sum_{i=1}^m\lambda_i f(z_i)\right|
    &\;\le\;
    \sup_{f\in\mathcal F}|\Lambda_{Q_0}f(y_{Q_0})|
    +\sum_{j=1}^J
       \sup_{f\in\mathcal F}
       \left|
         \sum_{\substack{Q\in\mathcal D_{dj}\\\mu_Z(Q)>0}}
           \Lambda_Q\bigl(f(y_Q)-f(y_{Q^-})\bigr)
       \right| \\
 &\qquad\qquad\qquad\qquad\quad
    +\sup_{f\in\mathcal F}
       \left|\sum_{i=1}^m\lambda_i
         \bigl(f(z_i)-f(y_J(z_i))\bigr)\right|.
\end{align*}

We now bound the three types of suprema separately.  Since
$\|f\|_\infty\le\sqrt d$ for every $f\in\mathcal F$,
\[
    \frac1m\E\sup_{f\in\mathcal F}
       |\Lambda_{Q_0}f(y_{Q_0})|
    \le\frac{\sqrt d}{m}\E|\Lambda_{Q_0}|
    \le\sqrt{\frac{2d}{m}}.
\]
For a fixed $j$, another application of the triangle inequality gives
\begin{align*}
 \sup_{f\in\mathcal F}
   \left|
      \sum_{\substack{Q\in\mathcal D_{dj}\\\mu_Z(Q)>0}}
        \Lambda_Q\bigl(f(y_Q)-f(y_{Q^-})\bigr)
    \right|
 & \le
    \sum_{\substack{Q\in\mathcal D_{dj}\\\mu_Z(Q)>0}}
       |\Lambda_Q|
       \sup_{f\in\mathcal F}|f(y_Q)-f(y_{Q^-})| \\
 &\le
    \sqrt d\,2^{-(j-1)}
    \sum_{\substack{Q\in\mathcal D_{dj}\\\mu_Z(Q)>0}}
       |\Lambda_Q|.
\end{align*}
Consequently, using the bound on $\E|\Lambda_Q|$ above and the definition of $N_{dj}^{\mathrm{eff}}(Z)$,
\begin{align*}
 \frac1m\E\sup_{f\in\mathcal F}
   \left|
      \sum_{\substack{Q\in\mathcal D_{dj}\\\mu_Z(Q)>0}}
        \Lambda_Q\bigl(f(y_Q)-f(y_{Q^-})\bigr)
    \right|
 & \le
    \frac{\sqrt d\,2^{-(j-1)}}m
       \sum_{\substack{Q\in\mathcal D_{dj}\\\mu_Z(Q)>0}}
          \sqrt{2m\mu_Z(Q)} \\
 & =
    \frac{\sqrt{2d}}{\sqrt m}\,
       2^{-(j-1)}\sqrt{N_{dj}^{\mathrm{eff}}(Z)}.
\end{align*}
Finally, the same pointwise argument and $\E|\lambda_i|=1$ give
\begin{align*}
 \frac1m\E\sup_{f\in\mathcal F}
    \left|\sum_{i=1}^m\lambda_i
      \bigl(f(z_i)-f(y_J(z_i))\bigr)\right|
 \le
    \frac{\sqrt d\,2^{-J}}m
       \sum_{i=1}^m\E|\lambda_i|
    =\sqrt d\,2^{-J}.
\end{align*}
Combining the root, increment, and remainder bounds proves \Cref{eq:laplacian-multiscale-bound}.
\end{proof}

\begin{corollary}[Laplacian complexity with intrinsic dimension]
\label{cor:intrinsic-laplacian-complexity}
Suppose that $A\ge1$, $0<s\le d$, and $A<m$.  Set
\[
    K_s\coloneqq\left\lceil\frac d s\log_2\!\left(\frac mA\right)\right\rceil
\]
and assume that
\begin{equation}\label{eq:laplacian-power-law-profile}
    N_k^{\mathrm{eff}}(Z)\le A2^{sk/d}
    \qquad\forall\,0\le k<K_s.
\end{equation}
Then
\begin{equation}\label{eq:intrinsic-laplacian-rates}
    \mathfrak L(\mathcal F;Z)
    \le C_s\sqrt d
    \begin{cases}
       (A/m)^{1/2}, & 0<s<2,\\
       (A/m)^{1/2}\bigl(1+\log(m/A)\bigr), & s=2,\\
       (A/m)^{1/s}, & 2<s\le d.
    \end{cases}
\end{equation}
\end{corollary}

\begin{proof}
Set
\[
    J_s\coloneqq\left\lceil\frac1s\log_2\!\left(\frac mA\right)\right\rceil-1.
\]
Since $dJ_s<K_s$, applying \Cref{prop:data-dependent-laplacian-complexity} and \Cref{eq:laplacian-power-law-profile} shows that the middle term in \Cref{eq:laplacian-multiscale-bound} is at most
\[
    C\sqrt{\frac{dA}{m}}\sum_{j=1}^{J_s}2^{(s/2-1)j},
\]
while $2^{-J_s}\le2(A/m)^{1/s}$ and $m^{-1/2}\le(A/m)^{1/2}$.  The geometric sum is bounded by a constant depending only on $s$ when $s<2$, has order $1+\log(m/A)$ when $s=2$, and is at most $C_s(A/m)^{1/s-1/2}$ when $s>2$.  This proves \Cref{eq:intrinsic-laplacian-rates}.
\end{proof}

\begin{remark}[Dependence on $s$]\label{rem:intrinsic-laplacian-s-dependence}
The proof shows that one may take
\[
    C_s\le C\left(1+\frac1{|s-2|}\right),
    \qquad s\ne2,
\]
and $C_2\le C$, where $C$ is universal.  The apparent blow-up as $s\to2$ comes from the geometric sum with ratio $2^{s/2-1}$; at the critical value $s=2$, its truncated length gives the logarithmic factor in \Cref{eq:intrinsic-laplacian-rates}.
\end{remark}

The $m^{-1/2}$ floor is unavoidable for the present function class.  Indeed, the constant function $f\equiv\sqrt d$ belongs to $\mathcal F$, so $\mathfrak L(\mathcal F;Z)\gtrsim\sqrt{d/m}$ for every $Z$.  Thus lower-dimensional structure improves the ambient $m^{-1/d}$ rate to $m^{-1/s}$ when $s>2$, while the rate saturates at the parametric scale below dimension $2$.  For an approximately uniform $d$-dimensional grid, $N_k^{\mathrm{eff}}(Z)$ grows like $2^k$ until it reaches $m$, so $s=d$ and the ambient rate is recovered.  More concentrated or uneven data can have a smaller effective occupancy and hence a better bound.

\section{Warm-up: Adaptive binary partition}\label{sec:warmup}

In optimal transportation theory, a hierarchical partitioning of the space is common. However, the uniform partition, even though it gives the optimal rate in the worst-case scenario, will often lead to the curse of high dimensionality (e.g. the Wasserstein error rate is $O(n^{-1/d})$ \cite{he2023algorithmically,boedihardjo2024private}). Therefore, we will take advantage of the intrinsic structure of the input data to improve the dimension dependence.

In this section, as a warm-up to the adaptive partitioning, we consider some naive algorithms without differential privacy requirement. Their models are relatively simple, but the analysis is instructive for understanding the private partition algorithms.

\subsection{Naive adaptive binary partition}

Fix a threshold $\theta\in(0,1]$.  Consider the following adaptive partitioning algorithm, which bisects any cell whose mass exceeds $\theta$ under a ground measure $\mu$.

\begin{algorithm}[H]
\caption{Naive adaptive binary partition}\label{alg:naive-partition}
\begin{algorithmic}[1]
\Require Probability measure $\mu$ on $[0,1]^d$ and a threshold $\theta\in(0,1]$
\State Initialize $\mathcal T_\theta$ with its root $\Omega$.
\State Initialize a queue $\mathcal H$ with $\mathcal H\gets\{\Omega\}$.
\While{$\mathcal H\ne\varnothing$}
    \State Pop the front cell $Q$ from $\mathcal H$.
    \If{$\mu(Q)>\theta$}
        \State Bisect $Q$ by the cyclic coordinate rule, obtaining $Q_0$ and $Q_1$, and add them to $\mathcal T_\theta$.
        \State Append $Q_0$ and $Q_1$ to the back of $\mathcal H$.
    \EndIf
\EndWhile
\State \Return $\mathcal T_\theta$.
\end{algorithmic}
\end{algorithm}

Assume for now that \Cref{alg:naive-partition} terminates, although this is not guaranteed in general. Then the leaves of $\mathcal T_\theta$ form a partition of $\Omega$. A capped version that forces termination by limiting the number of leaves is given in \Cref{app:capped-adaptive-partition}.
Apply \Cref{alg:exact-synthetic-measure} to $\mu$ and $\mathcal T_\theta$, and abbreviate its output by $\nu\coloneqq\nu_{\mathcal T_\theta}$.

\begin{theorem}
We have
\[
    W_1(\mu,\nu)\le 2\sqrt d\cdot \theta^{1/d}.
\]
Consequently, suppose that $\mu=\mu_n$ is the empirical measure associated with $n$ data points, and no point appears more than $c$ times in the data. Then \Cref{alg:naive-partition} terminates for $\theta=c/n$ and
\[
    W_1(\mu_n,\nu)
    \le 2\sqrt d\,c^{1/d}n^{-1/d}.
\]
\end{theorem}

\begin{proof}
Every terminal cell has $p_Q\le\theta$.
Thus \Cref{lem:leafwise-holder} and Kraft's identity \Cref{eq:kraft} give directly
\[
    \mathfrak R_{\mathcal T_\theta}(\mu)\le\theta^{1/d}.
\]
The resolution-error bound in \Cref{lem:resolution-error} proves the claim.
For the empirical consequence, at a sufficiently large depth every cell contains copies of at most one distinct data point.  Its mass is then at most $c/n=\theta$, so no further cell is split. 
\end{proof}

Here the accuracy decreases with the threshold $\theta$, but meanwhile the number of pieces in the partition increases since we are refining it. Moreover, when $\theta$ is too small, the algorithm may not terminate.

\subsection{Intrinsic-dimensional bounds}

By \Cref{lem:resolution-error}, bounding $W_1$ reduces to controlling the resolution functional $\mathfrak R_{\mathcal T}(\mu)$.  We next express it through the stopping-depth distribution to capture the intrinsic dimension of $\mu$.
For the remainder of this subsection, for simplicity, let $\mathcal T$ be the naive adaptive partition tree from \Cref{alg:naive-partition} and assume that it terminates at every point.
Then the threshold rule gives
\[
    K_{\mathcal T}(x)=\inf\{k\ge0:\mu(Q_k(x))\le\theta\}.
\]
Define the stopping-depth distribution function, which is the cumulative distribution function of the stopping depth $K_{\mathcal T}(x)$ under $\mu$, by
\begin{equation}\label{eq:stopping-depth-distribution}
    F_{\mu,\mathcal T}(k)
    \coloneqq\mu(\{x:K_{\mathcal T}(x)\le k\})
    =\sum_{\substack{Q\in\mathcal D_k\\0<\mu(Q)\le\theta}}\mu(Q).
\end{equation}
So $F_{\mu,\mathcal T}(k)$ is also the $\mu$-mass of the terminal cells of depth at most $k$ in the naive adaptive partition tree.

\begin{proposition}\label{prop:stopping-profile}
Under the preceding assumptions,
\begin{equation}\label{eq:stopping-depth-distribution-formula}
    W_1(\mu,\nu)
    \le2\sqrt d\,(1-\alpha)
       \sum_{k=0}^{\infty}\alpha^kF_{\mu,\mathcal T}(k).
\end{equation}
In particular,
\begin{equation}\label{eq:occupied-cell-entropy}
    W_1(\mu,\nu)
    \le2\sqrt d\,(1-\alpha)
       \sum_{k=0}^{\infty}\alpha^k
       \min\{1,\theta N_k^+(\mu)\}.
\end{equation}
\end{proposition}

\begin{proof}
Since the cells $Q_j(x)$ are nested, the sequence $\mu(Q_j(x))$ is nonincreasing in $j$.
Let $X\sim\mu$ and set $K=K_{\mathcal T}(X)$.

For any nonnegative, nonincreasing function $f$ on $\mathbb Z_{\ge0}$ satisfying $f(k)\to0$, the weighted tail-sum formula gives $\mathbb E[f(K)]=\sum_{k=0}^{\infty}\bigl(f(k)-f(k+1)\bigr)\mathbb P(K\le k)$.
When $f(k)=\alpha^k$, this gives
\[
    \mathfrak R_{\mathcal T}(\mu)
    =\mathbb E[\alpha^K]
    =(1-\alpha)\sum_{k=0}^{\infty}\alpha^k\mathbb P(K\le k)
    =(1-\alpha)\sum_{k=0}^{\infty}\alpha^kF_{\mu,\mathcal T}(k).
\]
Combining this identity with \Cref{lem:resolution-error} proves \Cref{eq:stopping-depth-distribution-formula}.
Finally, every summand in \Cref{eq:stopping-depth-distribution} is at most $\theta$, so
\[
    F_{\mu,\mathcal T}(k)\le\min\{1,\theta N_k^+(\mu)\}.
\]
This proves \Cref{eq:occupied-cell-entropy}.
\end{proof}

The exact value of $F_{\mu,\mathcal T}$ in \Cref{eq:stopping-depth-distribution-formula} is generally sharper than the occupied-cell bound with $N^+_k(\mu)$, because it retains the distribution of mass across cells.
Nevertheless, \Cref{eq:occupied-cell-entropy} provides a bound where the intrinsic dimension is explicit, as shown in the next corollary.

\begin{corollary}
\label{cor:power-law-covering}
Suppose that $A\theta<1$, $0<s\le d$, and \Cref{ass:finite-scale-intrinsic-dimension} holds for $\mu$ with $R=\theta^{-1}$.  Then
\begin{equation}\label{eq:intrinsic-rates}
    W_1(\mu,\nu)
    \le C_{d,s}
    \begin{cases}
       A\theta, & 0<s<1,\\
       A\theta\bigl(1+\log\frac{1}{A\theta}\bigr), & s=1,\\
       (A\theta)^{1/s}, & s>1,
    \end{cases}
\end{equation}
where $C_{d,s}$ depends only on $d$ and $s$.
\end{corollary}

\begin{proof}
Since $\alpha=2^{-1/d}$, the stated value of $K_s$ is precisely
\[
    K_s=\min\big\{k\ge0:\,\alpha^{k}\le (A\theta)^{1/s}\big\}.
\]
By \Cref{lem:resolution-error}, it suffices to control $\mathfrak R_{\mathcal T}(\mu)$. Applying \Cref{eq:occupied-cell-entropy} and \Cref{ass:finite-scale-intrinsic-dimension} and splitting the sum at $K_s$ gives
\begin{equation}\label{eq:split-covering-sum}
    \mathfrak R_{\mathcal T}(\mu)
    \le (1-\alpha)A\theta\sum_{k=0}^{K_s-1}\alpha^{(1-s)k}+\alpha^{K_s}.
\end{equation}
If $0<s<1$, then
\[
    \sum_{k=0}^{K_s-1}\alpha^{(1-s)k}
    \le\frac1{1-\alpha^{1-s}},
    \qquad
    \alpha^{K_s}\le (A\theta)^{1/s}\le A\theta,
\]
and concavity of $t\mapsto1-\alpha^t$ on $[0,1]$ gives
\[
    \frac{1-\alpha}{1-\alpha^{1-s}}\le\frac1{1-s}.
\]
Thus \Cref{eq:split-covering-sum} is at most $2A\theta/(1-s)$.
If $s=1$, then
\[
    K_s\le1+d\log_2\!\left(\frac1{A\theta}\right),
    \qquad \alpha^{K_s}\le A\theta.
\]
Since $d(1-\alpha)\le\log 2$, \Cref{eq:split-covering-sum} is at most $2A\theta\bigl(1+\log(1/(A\theta))\bigr)$.
Finally, if $s>1$, then
\[
    \sum_{k=0}^{K_s-1}\alpha^{(1-s)k}
    \le\frac{\alpha^{(1-s)K_s}}{\alpha^{1-s}-1},
    \qquad
    A\theta\,\alpha^{(1-s)K_s}\le \alpha^{1-s}(A\theta)^{1/s},
    \qquad
    \alpha^{K_s}\le (A\theta)^{1/s}.
\]
Substitution into \Cref{eq:split-covering-sum} gives
\[
    \mathfrak R_{\mathcal T}(\mu)
    \le (A\theta)^{1/s}\left(1+\frac{1-\alpha}{1-\alpha^{s-1}}\right).
\]
The fraction in parentheses is at most $\max\{1,(s-1)^{-1}\}$, again by concavity when $1<s<2$ and by monotonicity when $s\ge2$.
Hence the entire parenthesis is at most $2\max\{1,(s-1)^{-1}\}$, proving the $s>1$ case of \Cref{eq:intrinsic-rates}.
\end{proof}

\begin{remark}[Dependence on $s$ near the critical dimension]
The proof shows that, without optimizing numerical factors, one may take
\begin{equation}
    C_{d,s}=4\sqrt d
    \begin{cases}
       \displaystyle \frac1{1-s}, & 0<s<1,\\[6pt]
       1, & s=1,\\[4pt]
       \displaystyle \max\left\{1,\frac1{s-1}\right\}, & s>1.
    \end{cases}
\end{equation}
The apparent blow-up as $s\to1$ comes from bounds that discard the truncation of the geometric sum in \Cref{eq:split-covering-sum}. For $s\ne1$ near $1$, the factor $|s-1|^{-1}$ can be replaced, up to a universal constant, by
\[
    1+\min\left\{\frac1{|s-1|},\log\frac1{A\theta}\right\}.
\]
Consequently, one can obtain a bound uniform as $s\to1$ by paying at most one logarithmic factor:
\[
    W_1(\mu,\nu)
    \le C_d\left(1+\log\frac1{A\theta}\right)
    \begin{cases}
       A\theta, & s\le1,\\
       (A\theta)^{1/s}, & s>1,
    \end{cases}
    \qquad \frac12\le s\le2.
\]
\end{remark}

\begin{example}[Empirical data]
Let $\mu=\mu_n$ be the empirical measure associated with $n$ data points.  Fix constants $c\ge1$, $A\ge1$ and set $\theta\coloneqq c/n$.
Suppose that $1<s\le d$ and \Cref{ass:finite-scale-intrinsic-dimension} holds for $\mu_n$ with $R=\theta^{-1}=n/c$.  Assume the naive adaptive partition algorithm terminates with output $\mathcal T$.
Then \Cref{cor:power-law-covering} with $A\theta=Ac/n$ gives
\[
    W_1(\mu_n,\nu_{\mathcal T})
    \le4\sqrt d\max\left\{1,\frac1{s-1}\right\}
       \left(\frac{Ac}{n}\right)^{1/s}.
\]

\end{example}


\section{Partition with soft thresholds}\label{sec:soft-privtree}

A depth-dependent threshold gives a way to ensure termination while allowing us to introduce private noise in the partitioning later.
Partition with soft thresholds can be viewed as the deterministic skeleton of the partition in PrivTree \cite{zhang2016privtree}: the mechanism compares the cell mass with the depth-dependent threshold $\theta+\depth(Q)\delta$.
We therefore split a cell $Q$ when
\[
    \mu(Q)>\theta+\depth(Q)\delta,
\]
where $\theta\in(0,1]$ and $\delta>0$.
We first consider the deterministic partition and corresponding resolution error, and later we will discuss the noisy split decisions in PrivTree.

\begin{algorithm}[H]
\caption{Adaptive binary partition with a soft threshold}
\label{alg:soft-threshold-partition}
\begin{algorithmic}[1]
\Require Probability measure $\mu$ on $[0,1]^d$, base threshold $\theta\in(0,1]$, and increment $\delta>0$
\State Initialize $\mathcal T_{\theta,\delta}$ with its root $\Omega$.
\State Initialize a queue $\mathcal H$ with $\mathcal H\gets\{\Omega\}$.
\While{$\mathcal H\ne\varnothing$}
    \State Pop the front cell $Q$ from $\mathcal H$ and set $k\gets\depth(Q)$.
    \If{$\mu(Q)>\theta+k\delta$}
        \State Bisect $Q$ by the cyclic coordinate rule, obtaining $Q_0$ and $Q_1$, and add them to $\mathcal T_{\theta,\delta}$.
        \State Append $Q_0$ and $Q_1$ to the back of $\mathcal H$.
    \EndIf
\EndWhile
\State \Return $\mathcal T_{\theta,\delta}$.
\end{algorithmic}
\end{algorithm}

\subsection{Termination, size, and Wasserstein error}

The increasing threshold makes termination automatic, including for measures with atoms.
This is the main qualitative difference from \Cref{alg:naive-partition}.

\begin{proposition}[Termination]
\label{prop:soft-threshold-termination}
Let $k_{\max}\coloneqq\left\lceil\frac{1-\theta}{\delta}\right\rceil$.
Then \Cref{alg:soft-threshold-partition} has depth at most $k_{\max}$ and hence terminates after finitely many splits.
More precisely, if $\theta<1$, then
\begin{equation}\label{eq:soft-number-terminal-cells}
\begin{split}
    \#\{\text{terminal cells}\}
    &\le 1+\sum_{k=0}^{k_{\max}-1}
       \min\left\{2^k,\frac1{\theta+k\delta}\right\}
    \le 1+\frac1\theta+\frac1\delta\log\left(\frac1\theta\right).
\end{split}
\end{equation}
\end{proposition}

\begin{proof}
At every depth $k\ge k_{\max}$, the threshold $\theta+k\delta\geq 1$.
Since $\mu(Q)\le1$, no such cell is split, which proves finite termination and the depth bound.

Note that the number of terminal cells is one more than the number of split cells.
Let $s_k$ be the number of split cells at depth $k$.
Such cells are disjoint and each has mass strictly larger than $\theta+k\delta$, so
\[
    s_k\le \min\left\{2^k,\frac1{\theta+k\delta}\right\}.
\]
Summing over the possible splitting depths gives the first inequality in \Cref{eq:soft-number-terminal-cells}.
The summand $(\theta+k\delta)^{-1}$ is decreasing, and every possible splitting depth satisfies $k<(1-\theta)/\delta$.
Therefore,
\[
    \sum_{\substack{0\leq  k<\frac{1-\theta}{\delta}}}
       \frac1{\theta+k\delta}
    \le \frac1\theta+
       \int_0^{(1-\theta)/\delta}\frac{dx}{\theta+x\delta}
    \le \frac1\theta+\frac1\delta\log\left(\frac1\theta\right).
\]
This proves the second inequality in \Cref{eq:soft-number-terminal-cells}.
\end{proof}

Having established termination and size, we next bound the Wasserstein error.
Let $\mathcal T=\mathcal T_{\theta,\delta}$ be the output of \Cref{alg:soft-threshold-partition} and let $\nu_{\theta,\delta}\coloneqq\nu_{\mathcal T}$ be the exact synthetic measure produced by \Cref{alg:exact-synthetic-measure}.
Write $\mathcal L=\mathcal L(\mathcal T)$ and denote
\[
    L_\theta\coloneqq\left\lceil\log_2\left(\frac1\theta\right)\right\rceil.
\]
The next theorem shows that the price of the increasing threshold is only a logarithmic factor multiplying $\delta$.

\begin{theorem}[Soft-threshold resolution error]
\label{thm:w1-soft-threshold}
For every probability measure $\mu$ on $[0,1]^d$, the exact synthetic measure associated with the output tree of \Cref{alg:soft-threshold-partition} satisfies
\begin{equation}
    W_1(\mu,\nu_{\theta,\delta})
    \le 2\sqrt d\,
       \left[\theta+\delta(L_\theta+2)\right]^{1/d}.
\end{equation}
\end{theorem}

\begin{proof}
By \Cref{lem:resolution-error,lem:leafwise-holder}, it suffices to prove
\begin{equation}
    \label{eq:R_bound_soft}
    \sum_{Q\in\mathcal L}p_Q2^{-\depth(Q)}
    \le \theta+\delta(L_\theta+2).
\end{equation}
Every leaf $Q$ satisfies $p_Q\le\theta+\depth(Q)\delta$.  Let $U$ be uniform on $[0,1]^d$. Since $\mathcal L$ is a partition of $[0,1]^d$, the tail-sum formula therefore gives
\[
\begin{aligned}
    \sum_{Q\in\mathcal L}\depth(Q)2^{-\depth(Q)}
    =\mathbb E[K_{\mathcal T}(U)]
    & = \sum_{k=0}^{\infty}\mathbb P(K_{\mathcal T}(U)>k) \\
    & \le\sum_{k=0}^{\infty}\min\left\{1,\frac{2^{-k}}\theta\right\}
    \le L_\theta+2.
\end{aligned}
\]
The first inequality holds because the event $\{K_{\mathcal T}(U)>k\}$ is the union of the split cells at depth $k$, of which there are at most $1/\theta$.
The last inequality follows by splitting the sum at $L_\theta$: the first $L_\theta$ terms are at most $1$, while $2^{-L_\theta}\le\theta$ implies $\theta^{-1}\sum_{k=L_\theta}^{\infty}2^{-k}\le2$.

Substituting this bound into the left-hand side of \Cref{eq:R_bound_soft} gives the desired result.
\end{proof}

\paragraph{Empirical measures and intrinsic dimension.}
For the soft-threshold tree, monotonicity gives $K_{\mathcal T}(x)\le k$ if and only if $\mu(Q_k(x))\le\theta+k\delta$.  Thus the stopping-depth distribution is
\begin{equation}\label{eq:soft-stopping-profile}
    F_{\mu,\mathcal T}(k)
    =\sum_{\substack{Q\in\mathcal D_k\\
          0<\mu(Q)\le\theta+k\delta}}\mu(Q).
\end{equation}
The weighted tail-sum identity from the proof of \Cref{prop:stopping-profile}, together with \Cref{eq:soft-stopping-profile}, gives
\begin{equation}\label{eq:soft-stopping-profile-bound}
\begin{split}
    \mathfrak R_{\mathcal T}(\mu)
    &=(1-\alpha)
       \sum_{k=0}^{\infty}\alpha^kF_{\mu,\mathcal T}(k) \\
    &\le(1-\alpha)
       \sum_{k=0}^{\infty}\alpha^k
       \min\{1,(\theta+k\delta)N_k^+(\mu)\}.
\end{split}
\end{equation}

\begin{corollary}
Let $\mu=\mu_n$ be the empirical measure associated with $n$ data points where $n$ is sufficiently large. Fix $b,c>0$ as absolute constants, set
\begin{equation}
    \theta=\frac{c\log n}{n},
    \qquad
    \delta=\frac bn,
\end{equation}
and let $\mathcal T$ be the output of \Cref{alg:soft-threshold-partition}.
\begin{enumerate}
    \item There is
    \[
        W_1(\mu_n,\nu_{\mathcal T})
        \le C\sqrt d\left(\frac{\log n}{n}\right)^{1/d},
        \qquad
        |\mathcal L(\mathcal T)|\le Cn\log n.
    \]

    \item Moreover, if for some fixed $A\ge1$ and $1<s\le d$, \Cref{ass:finite-scale-intrinsic-dimension} holds for $\mu_n$ with $R=n/\log n$, then
    \[
        W_1(\mu_n,\nu_{\mathcal T})
        \le C_{d,s}\left(\frac{A\log n}{n}\right)^{1/s}.
    \]
\end{enumerate}
\end{corollary}

\begin{proof}
Since $b$ and $c$ are absolute constants, $L_\theta\le C\log n$ and
\[
    \theta+\delta(L_\theta+2)\le C\frac{\log n}{n}.
\]
Thus \Cref{thm:w1-soft-threshold} gives the first bound in part~1, with the displayed dependence on $d$. Moreover, \Cref{eq:soft-number-terminal-cells} gives
\[
    |\mathcal L(\mathcal T)|
    \le 1+\frac1\theta+\frac1\delta\log\left(\frac1\theta\right)
    \le Cn\log n.
\]
For part~2, if $k<K_s$, then $k=O_{d,s}(\log n)$ and hence
\[
    (\theta+k\delta)N_k^+(\mu_n)
    \le C_{d,s}\frac{A\log n}{n}\alpha^{-sk}.
\]
For $k\ge K_s$, the minimum in \Cref{eq:soft-stopping-profile-bound} is at most $1$, so the tail is bounded by $(1-\alpha)\sum_{k=K_s}^{\infty}\alpha^k=\alpha^{K_s}$.  Splitting the sum at $K_s$ therefore gives
\[
    \mathfrak R_{\mathcal T}(\mu_n)
    \le C_{d,s}(1-\alpha)\frac{A\log n}{n}
       \sum_{k=0}^{K_s-1}\alpha^{(1-s)k}+\alpha^{K_s}
    \le C_{d,s}\left(\frac{A\log n}{n}\right)^{1/s},
\]
as in the $s>1$ case of \Cref{cor:power-law-covering}. The second assertion follows from \Cref{lem:resolution-error}, with its factor $2\sqrt d$ absorbed into $C_{d,s}$.
\end{proof}

\subsection{Accuracy of the PrivTree partition}

We now write the noisy split decisions of PrivTree \cite[Algorithm~2]{zhang2016privtree} and consider the new resolution error after the private noise is introduced.

\begin{algorithm}[H]
\caption{Binary PrivTree partition}
\label{alg:binary-privtree}
\begin{algorithmic}[1]
\Require Empirical measure $\mu_n$ on $[0,1]^d$, base threshold $\theta\in[0,1)$, increment $\delta>0$, and noise scale $\sigma>0$
\State Initialize $\mathcal T_{\priv}$ with its root $\Omega$.
\State Initialize a queue $\mathcal H$ with $\mathcal H\gets\{\Omega\}$.
\While{$\mathcal H\ne\varnothing$}
    \State Pop the front cell $Q$ from $\mathcal H$ and set $k\gets\depth(Q)$.
    \State Set $p_Q\gets\mu_n(Q)$.
    \State Draw an independent $\xi_Q\sim\operatorname{Lap}(\sigma)$ and set
    \[\widehat p_Q\gets\max\{\theta+(k-1)\delta,p_Q\}+\xi_Q.\]
    \If{$\widehat p_Q>\theta+k\delta$}
        \State Bisect $Q$ by the cyclic coordinate rule, obtaining $Q_0$ and $Q_1$, and add them to $\mathcal T_{\priv}$.
        \State Append $Q_0$ and $Q_1$ to the back of $\mathcal H$.
    \EndIf
\EndWhile
\State \Return $\mathcal T_{\priv}$ with all cell masses removed.
\end{algorithmic}
\end{algorithm}

Let $\nu_{\mathcal T_{\priv}}$ denote the oracle synthetic measure obtained by applying \Cref{alg:exact-synthetic-measure} to $\mu_n$ and $\mathcal T_{\priv}$. Note that it uses the exact leaf masses and is not itself a private release.

\begin{proposition}[Termination of PrivTree]
\label{prop:privtree-termination}
\Cref{alg:binary-privtree} terminates after finitely many splits almost surely.
Moreover, if $m=|\mathcal L(\mathcal T_{\priv})|$, then
\begin{equation}
    \mathbb E[m]
    \le (4+\delta^{-1})
       \frac{1-\frac12e^{-\delta/\sigma}}
            {1-e^{-\delta/\sigma}}.
\end{equation}
\end{proposition}

\begin{proof}
For every cell $Q\in \mathcal D_k$, set $p_Q=\mu_n(Q)$ and call it heavy if $p_Q>\theta+(k-1)\delta$.  There are at most 3 heavy cells at depths zero and one.  To count those at greater depths, fix $x$ and let $Q_K(x)$ where $K\ge2$ be the deepest heavy cell on its path, if one exists.  All its ancestors are heavy, and
\[
    \sum_{\substack{k\ge2:\ Q_k(x)\text{ is heavy}}}
       \frac1{p_{Q_k(x)}}
    \le \frac{K-1}{p_{Q_K(x)}}
    <\frac1\delta.
\]
Since every heavy cell has positive mass, integration with respect to $\mu_n$ gives
\[
    \#\{Q:\depth(Q)\ge2,\ Q\text{ is heavy}\}
    =\sum_{\substack{k\ge2:\ Q_k(x)\text{ is heavy}}}\int_\Omega
          \frac1{p_{Q_k(x)}}\,d\mu_n(x)
    \le\frac1\delta,
\]
where in the last step we exchanged the order of summation and integration and used the uniform upper bound for the summation.
Thus the full hierarchy contains at most $H=3+\delta^{-1}$ heavy cells.

A non-heavy cell is split precisely when $\xi_Q>\delta$, and all its descendants are non-heavy.  Set
\[
    q\coloneqq\mathbb P\{\xi_Q>\delta\}
      =\frac12e^{-\delta/\sigma},
    \qquad 2q<1.
\]
We next consider a non-heavy subtree of $\mathcal T_{\priv}$ rooted at a non-heavy cell.  Let $\ell_r$ be the expected number of leaves in a non-heavy subtree truncated at relative depth $r$.  Conditioning on the root split gives
\[
    \ell_0=1,
    \qquad
    \ell_{r+1}=1-q+2q\ell_r.
\]
Therefore,
\[
    \ell_r
    =(1-q)\sum_{j=0}^{r-1}(2q)^j+(2q)^r
    \longrightarrow\frac{1-q}{1-2q}.
\]
Monotone convergence shows that a non-heavy subtree terminates almost surely and has expected number of leaves $(1-q)/(1-2q)$.

Stop each branch when it first reaches a non-heavy cell.  The resulting binary tree $\mathcal T'$ has at most $H$ internal nodes (as all internal nodes are heavy cells), and hence at most $H+1$ leaves.  Each of these leaves of $\mathcal T'$ is either a terminal heavy cell, contributing one leaf to $\mathcal T_{\priv}$, or the root of a non-heavy subtree, contributing $(1-q)/(1-2q)$ leaves in expectation.  Therefore,
\[
    \mathbb E[m]
    \le (H+1)\max\left\{1,\frac{1-q}{1-2q}\right\}
    =(4+\delta^{-1})
       \frac{1-\frac12e^{-\delta/\sigma}}
            {1-e^{-\delta/\sigma}}.
\]
This finite expectation also proves almost-sure termination.
\end{proof}

\begin{theorem}[Resolution error of PrivTree]\label{thm:privtree-accuracy}
For every $t>0$ such that $\theta+t<1$, the aggregated measure $\nu_{\mathcal T_{\priv}}$ induced by the output of \Cref{alg:binary-privtree} satisfies
\begin{equation}\label{eq:privtree-general-accuracy}
\begin{split}
    \mathbb E\bigl[W_1(\mu_n,\nu_{\mathcal T_{\priv}})\bigr]
    \le 2\sqrt d\Bigg\{
       \left[\theta+t+\delta(L_{\theta+t}+2)\right]^{1/d}
       +\frac12\left(1+\frac1\delta\right)e^{-t/\sigma}
       \Bigg\}.
\end{split}
\end{equation}
Here the expectation is over the noise used in the split decisions.
\end{theorem}

\begin{proof}
Let $K_{\priv}(x)$ be the depth of the PrivTree leaf containing $x$, and let $K_t(x)$ be the stopping depth of the deterministic soft-threshold tree $\mathcal T_{\theta+t,\delta}$.
If $K_{\priv}(x)<K_t(x)$, then at some depth $j<K_t(x)$ the PrivTree cell $Q_j(x)$ is not split even though
\[
    p_{Q_j(x)}>\theta+t+j\delta.
\]
At such a cell, the maximum in \Cref{alg:binary-privtree} gives $p_{Q_j(x)}$, so failure to split implies $p_{Q_j(x)}+\xi_{Q_j(x)}\le\theta+j\delta$ and hence $\xi_{Q_j(x)}<-t$.

There are at most $K_t(x)$ possible depths $j<K_t(x)$, and at each one the Laplace lower tail satisfies $\mathbb P\{\xi_{Q_j(x)}<-t\}=\frac12e^{-t/\sigma}$.
The union bound therefore gives
\begin{equation}\label{eq:privtree-premature-stopping}
    \mathbb P\{K_{\priv}(x)<K_t(x)\}
    \le \frac{K_t(x)}2e^{-t/\sigma}
    \le \frac12\left(1+\frac1\delta\right)e^{-t/\sigma},
\end{equation}
where the last inequality follows because $p_Q\le1$ implies $K_t(x)\le k_{\max}\le1+1/\delta$.
Since $k\mapsto\alpha^k$ is decreasing, pointwise we have
\begin{align*}
    \alpha^{K_{\priv}(x)}
    &\leq \alpha^{K_t(x)}\cdot \mathbf 1_{\{K_{\priv}(x)\geq K_t(x)\}}
       +\alpha^{K_{\priv}(x)}\mathbf 1_{\{K_{\priv}(x)<K_t(x)\}}\\
    &\leq \alpha^{K_t(x)}
       +\mathbf 1_{\{K_{\priv}(x)<K_t(x)\}}.
\end{align*}
Taking expectations and using \Cref{eq:privtree-premature-stopping} gives
\[
    \mathbb E\bigl[\alpha^{K_{\priv}(x)}\bigr]
    \le \alpha^{K_t(x)}
       +\frac12\left(1+\frac1\delta\right)e^{-t/\sigma}.
\]
After integration with respect to $\mu_n$ and applying Fubini's theorem, we have
\begin{align}
    \mathbb E\bigl[W_1(\mu_n,\nu_{\mathcal T_{\priv}})\bigr]
    &\le 2\sqrt d\int_\Omega
       \mathbb E\bigl[\alpha^{K_{\priv}(x)}\bigr]d\mu_n(x) \notag\\
    &\le 2\sqrt d\left(
       \mathfrak R_{\mathcal T_{\theta+t,\delta}}(\mu_n)
       +\frac12\left(1+\frac1\delta\right)e^{-t/\sigma}
       \right) \label{eq:privtree-resolution-comparison}\\
    &\le 2\sqrt d\left(
       \left[\theta+t+\delta(L_{\theta+t}+2)\right]^{1/d}
       +\frac12\left(1+\frac1\delta\right)e^{-t/\sigma}
       \right), \notag
\end{align}
where the first inequality applies \Cref{lem:resolution-error} to each realized tree, and the last uses the resolution-error estimate in the proof of \Cref{thm:w1-soft-threshold}.
\end{proof}

We now combine the private parameter choice of PrivTree with \Cref{thm:privtree-accuracy}.

\begin{corollary}\label{cor:binary-privtree-rate}
Let $0<\eps\le1$ and $\eps n>2$.  Set
\begin{equation}
    \theta=0,
    \qquad
    \sigma=\frac3{\eps n},
    \qquad
    \delta=\sigma\log 2,
\end{equation}
the normalized binary-tree choice from \cite[Corollary~1]{zhang2016privtree}.
Then \Cref{alg:binary-privtree} returns an $\eps$-differentially private tree $\mathcal T_{\priv}$.
Moreover, for a universal constant $C$,
\begin{equation}
    \mathbb E\bigl[W_1(\mu_n,\nu_{\mathcal T_{\priv}})\bigr]
    \le C\sqrt d\left(\frac{\log(\eps n)}{\eps n}\right)^{1/d}.
\end{equation}
\end{corollary}

\begin{proof}
For a binary tree, the noise scale $3/\eps$ in \cite[Corollary~1]{zhang2016privtree} becomes $3/(\eps n)$ after normalizing counts by $n$, which proves the privacy claim under the same add/remove neighboring relation.
If $\eps n\ge18$, then $6\log(\eps n)<\eps n$, so we may apply \Cref{thm:privtree-accuracy} with $t=2\sigma\log(\eps n)<1$.  Then
\[
    \theta+t=2\sigma\log(\eps n),
    \qquad
    \delta=\sigma\log 2,
    \qquad
    L_{\theta+t}\le\left\lceil\log_2(\eps n)\right\rceil,
\]
and hence
\[
    \theta+t+\delta(L_{\theta+t}+2)\le C\frac{\log(\eps n)}{\eps n},
    \qquad
    \left(1+\frac1\delta\right)e^{-t/\sigma}\le\frac{C}{\eps n}.
\]
Substitution into \Cref{eq:privtree-general-accuracy} proves the claimed estimate when $\eps n\ge18$.  If $2<\eps n<18$, then $\log(\eps n)/(\eps n)\ge\log 18/18$, so the same estimate follows from $W_1\le\diam(\Omega)=\sqrt d$ after increasing the universal constant $C$.
\end{proof}

\begin{theorem}[Private measure from PrivTree]
\label{thm:privtree_1/d}
Let $0<\eps\le1$ and $\eps n>2$.  Run \Cref{alg:binary-privtree} with
\[
    \theta=0,
    \qquad
    \sigma=\frac6{\eps n},
    \qquad
    \delta=\sigma\log2,
\]
and then \Cref{alg:private-synthetic-measure} with budget $\eps/2$.
Then $\mathbb E\bigl[|\mathcal L(\mathcal T_{\priv})|\bigr]\le C\eps n$.
The output measure $\widehat\nu_{\mathcal T_{\priv}}$ is $\eps$-differentially private and satisfies
\begin{equation}\label{eq:fully-private-privtree-error}
    \mathbb E\bigl[W_1(\mu_n,\widehat\nu_{\mathcal T_{\priv}})\bigr]
    \le C\sqrt d
    \begin{cases}
        (\eps n)^{-1/2}, & d=1,\\
        (\eps n)^{-1/2}\log(\eps n), & d=2,\\
        \left(\dfrac{\log(\eps n)}{\eps n}\right)^{1/d},
            & d\ge3.
    \end{cases}
\end{equation}
Here $C$ is universal, and the expectation is over the noise used in the split decisions and the leaf-mass release.
\end{theorem}

\begin{proof}
With the stated noise scale, $\mathcal T_{\priv}$ is $\eps/2$-differentially private by \cite[Corollary~1]{zhang2016privtree}.  For every fixed tree, the leaf-mass release is $\eps/2$-differentially private by \Cref{prop:private-leaf-mass-error}.  Hence \Cref{lem:adaptive-composition} proves that $\widehat\nu_{\mathcal T_{\priv}}$ is $\eps$-differentially private.

The accuracy bound in \Cref{cor:binary-privtree-rate}, with privacy budget $\eps/2$, gives
\[
    \mathbb E\bigl[W_1(\mu_n,\nu_{\mathcal T_{\priv}})\bigr]
    \le C\sqrt d\left(\frac{\log(\eps n)}{\eps n}\right)^{1/d}.
\]

Let $m=|\mathcal L(\mathcal T_{\priv})|$.
Since $\delta=\sigma\log2=6\log2/(\eps n)$, \Cref{prop:privtree-termination} gives
\[
    \mathbb E[m]\le 6+\frac3{2\delta}
    =6+\frac{\eps n}{4\log2}
    \le C\eps n.
\]
By Jensen's inequality,
\[
    \mathbb E[m^{1/2}]\le(\mathbb E[m])^{1/2},
    \qquad
    \mathbb E[m^{1-1/d}]\le(\mathbb E[m])^{1-1/d}\quad(d\ge3).
\]
For $d=2$, the function $x\mapsto\sqrt{x}\log(ex)$ is concave on $[1,\infty)$ and dominates $\sqrt{x}\log(1+x)$, so Jensen's inequality also gives
\[
    \mathbb E\bigl[m^{1/2}\log(1+m)\bigr]
    \le (\mathbb E[m])^{1/2}\log(e\mathbb E[m]).
\]
Since $\eps n>2$, the right-hand sides above are bounded respectively by $C(\eps n)^{1/2}$, $C(\eps n)^{1/2}\log(\eps n)$, and $C(\eps n)^{1-1/d}$.
Finally, condition on $\mathcal T_{\priv}$ and apply \Cref{prop:private-leaf-mass-error}.  Averaging over the tree, using the preceding moment bounds, the triangle inequality, and the resolution-error bound above gives the sum of the resolution and leaf-mass terms; the factor $2$ from the leaf-mass budget $\eps/2$ is absorbed into $C$.  The resolution term is absorbed by the leaf-mass term when $d=1,2$; when $d\ge3$, the leaf-mass term is absorbed by the resolution term.  This proves \Cref{eq:fully-private-privtree-error}.
\end{proof}

\subsection{Intrinsic dimension dependence of PrivTree}

We finally combine the intrinsic-dimensional partition bound with the data-dependent Laplacian complexity bound.  We aim to avoid the curse of high dimensionality and deduce a 1-Wasserstein error bound of order $(\eps n)^{-1/s}$ up to logarithmic factors.
Throughout this subsection, assume $d\ge3$ and $2<s\le d$. This is because the Laplacian complexity saturates at the parametric rate when $s\le2$, as noted after \Cref{cor:intrinsic-laplacian-complexity}.

There is one additional issue in applying \Cref{prop:data-dependent-laplacian-complexity}.  Conditional on the private tree, that proposition must be applied to the representatives
\[
    Z
    \coloneqq (y_Q)_{Q\in\mathcal L(\mathcal T_{\priv})},
\]
rather than to the original data.  Noise may create branches inside empty cells, so a low-dimensional bound for the original data does not transfer to every realization of $Z$.  The next lemma gives the expected form of the transfer that we need.

\begin{lemma}[Leaf representatives of PrivTree]
\label{lem:private-representative-occupancy}
Run \Cref{alg:binary-privtree} with $\theta=0$ and $\delta=\sigma\log2$.  Fix $A\ge1$, $0<s\le d$, and $R>A$, and suppose that \Cref{ass:finite-scale-intrinsic-dimension} holds for $\mu_n$ with parameters $(A,s,R)$.
For any choice of representatives $y_Q\in Q$, the data set $Z$ defined above satisfies
\begin{equation}\label{eq:private-representative-occupancy}
    \mathbb E\bigl[N_k^{\mathrm{eff}}(Z)\bigr]
    \le
    \mathbb E\bigl[N_k^+(\mu_Z)\bigr]
    \le \frac{6A}{2^{s/d}-1}\,2^{sk/d}
    \qquad\forall\,0\le k\le K_s,
\end{equation}
where the expectation is over the split-decision noise.
\end{lemma}

\begin{proof}
For a realization of the private tree $\mathcal T_\priv$, let $V_j$ be the number of its nodes at depth $j$.  A depth-$j$ node is called heavy if its mass is larger than $(j-1)\delta$, as in the proof of \Cref{prop:privtree-termination}.  Let $H_j$ be the number of heavy nodes with depth $j$.  Every heavy node is non-empty, so
\[
    H_j\le N_j^+(\mu_n).
\]
A non-heavy node is split exactly when its noise exceeds $\delta$.  Under the stated parameter choice, this event has probability
\[
    q=\frac12e^{-\delta/\sigma}=\frac14.
\]
Conditioning on the nodes reached at depth $j$, and bounding the probability of splitting a heavy node by one, gives
\[
    \mathbb E[V_{j+1}\mid V_j,H_j]
    \le 2H_j+2q(V_j-H_j)
    =\frac12V_j+\frac32H_j.
\]
Consequently, using \Cref{ass:finite-scale-intrinsic-dimension},
\[
    \mathbb E[V_{j+1}]
    \le\frac12\mathbb E[V_j]+\frac32A2^{sj/d},
    \qquad \forall \; 0\le j<K_s.
\]
Since $V_0=1$ and $1<2^{s/d}\le2$, iterating this recurrence gives
\begin{equation}\label{eq:private-tree-level-size}
\begin{split}
    \mathbb E[V_j]
    &\le 2^{-j}
       +\frac32A\sum_{\ell=0}^{j-1}
          2^{-(j-1-\ell)}2^{s\ell/d}
    \leq
      \frac{3}{2(2^{s/d}-\frac12)}
      A2^{sj/d}
    \leq 3A2^{sj/d},
    \qquad \forall \; 0\le j\le K_s.
\end{split}
\end{equation}

It remains to pass from tree nodes to leaf representatives.  A leaf of depth at least $k$ has its representative in the depth-$k$ ancestor of that leaf, which is one of the $V_k$ tree nodes.  Each leaf $Q$ of depth less than $k$ can contribute at most one additional occupied depth-$k$ cell to which $y_Q$ belongs.  Therefore, deterministically,
\[
    N_k^+(\mu_Z)
    = \# \text{ leaves of $\mathcal T_\priv$ truncated after depth $k$}
    \le V_k+\sum_{j=0}^{k-1}V_j
    =\sum_{j=0}^kV_j.
\]
Taking expectations, applying \Cref{eq:private-tree-level-size}, and using
\[
    \sum_{j=0}^k2^{sj/d}
    \le\frac{2}{2^{s/d}-1}\,2^{sk/d}
\]
prove the second inequality in \Cref{eq:private-representative-occupancy}.  The first follows from \Cref{eq:effective-occupancy-comparison}.
\end{proof}

The lemma is an expectation bound rather than a statement for every realized tree.  This is unavoidable: false-positive splits inside empty cells form subcritical random binary trees whose sizes are finite almost surely but have no deterministic upper bound.  We therefore use the multiscale bound in \Cref{prop:data-dependent-laplacian-complexity} directly and average over the tree; the deterministic-profile corollary \Cref{cor:intrinsic-laplacian-complexity} cannot be applied to each realization with a common constant.

\begin{theorem}[PrivTree with intrinsic-dimensional accuracy]
\label{thm:fully-private-intrinsic-privtree}
Follow the algorithm in \Cref{thm:privtree_1/d} under the same assumptions.  We further assume that $d\ge3$, $2<s\le d$, $1\le A<\eps n$, and \Cref{ass:finite-scale-intrinsic-dimension} holds for $\mu_n$ with $R=\eps n$.
Then
\begin{equation}\label{eq:fully-private-intrinsic-privtree-rate}
    \mathbb E\bigl[W_1(\mu_n,\widehat\nu_{\mathcal T_{\priv}})\bigr]
    \le C_s\sqrt d\left(\frac{Ad\log(\eps n)}{\eps n}\right)^{1/s}.
\end{equation}
\end{theorem}

\begin{proof}
Recall that
\[\theta=0,\quad  \sigma=6/(\eps n), \quad \delta=\sigma\log2,\]
and the leaf-mass release receives budget $\eps/2$.  The claim is immediate from $W_1\le\sqrt d$ when $\eps n$ is bounded, so assume that $\eps n$ is sufficiently large.

For the resolution error, set $t=2\sigma\log(\eps n)$.  Since $k<K_s$ implies $t+k\delta\le C_s d\log(\eps n)/(\eps n)$, combining \Cref{eq:privtree-resolution-comparison,eq:soft-stopping-profile-bound} with \Cref{ass:finite-scale-intrinsic-dimension} and splitting the sum at $K_s$ gives
\begin{align*}
    \mathbb E\bigl[W_1(\mu_n,\nu_{\mathcal T_{\priv}})\bigr]
    &\le 2\sqrt d\Bigg\{(1-\alpha)\sum_{k=0}^{\infty}\alpha^k
       \min\{1,(t+k\delta)N_k^+(\mu_n)\}
       +\frac12\left(1+\frac1\delta\right)e^{-t/\sigma}\Bigg\} \\
    &\le 2\sqrt d\Bigg\{(1-\alpha)\sum_{k=0}^{K_s-1}\alpha^k
       \min\{1,(t+k\delta)N_k^+(\mu_n)\}
       +\alpha^{K_s}+\frac12\left(1+\frac1\delta\right)e^{-t/\sigma}\Bigg\} \\
    &\le C_s\sqrt d\Bigg\{(1-\alpha)\sum_{k=0}^{K_s-1}\alpha^k
       \min\left\{1,\frac{Ad\log(\eps n)}{\eps n}\alpha^{-sk}\right\}
       +\alpha^{K_s}+\frac1{\eps n}\Bigg\} \\
    &\le C_s\sqrt d\left(\frac{Ad\log(\eps n)}{\eps n}\right)^{1/s}.
\end{align*}
The last step is the same geometric-sum estimate as in the $s>1$ case of \Cref{cor:power-law-covering}, and also $\alpha^{K_s}\le(A/(\eps n))^{1/s}$.

It remains to control the leaf-mass error.  Let $m=|\mathcal L(\mathcal T_{\priv})|$ and $J=\lfloor K_s/d\rfloor$.  Conditional on the tree, \Cref{prop:private-leaf-mass-error,prop:data-dependent-laplacian-complexity} with privacy budget $\eps/2$ give
\begin{align*}
    \E\bigl[
       W_1(\nu_{\mathcal T_{\priv}},\widehat\nu_{\mathcal T_{\priv}})
       \mid\mathcal T_{\priv}\bigr]
    \le \frac{C\sqrt d}{\eps n}\left\{
       \sqrt m+\sum_{j=1}^J2^{-(j-1)}\sqrt{mN_{dj}^{\mathrm{eff}}(Z)}
       +m2^{-J}\right\}.
\end{align*}
Here the conditional expectation is taken over the leaf-mass noise.
By \Cref{prop:privtree-termination}, $\mathbb E[m]\le C\eps n$.  Since $N_{dj}^{\mathrm{eff}}(Z)\le m$, \Cref{lem:private-representative-occupancy} gives
\[
    \mathbb E\bigl[N_{dj}^{\mathrm{eff}}(Z)\bigr]
    \le C_s\min\{\eps n,dA2^{sj}\}.
\]
Consequently, Cauchy--Schwarz yields
\[
    \mathbb E\sqrt{mN_{dj}^{\mathrm{eff}}(Z)}
    \le \sqrt{\mathbb E[m]\,\mathbb E[N_{dj}^{\mathrm{eff}}(Z)]}
    \le C_s\sqrt{\eps n\min\{\eps n,dA2^{sj}\}}.
\]
Splitting the resulting geometric sum where $dA2^{sj}$ reaches $\eps n$ gives
\begin{align*}
    \mathbb E\bigl[W_1(\nu_{\mathcal T_{\priv}},
       \widehat\nu_{\mathcal T_{\priv}})\bigr]
    &\le C_s\sqrt d\left\{\frac{1}{\sqrt{\eps n}}
       +\sum_{j=1}^J2^{-j}\min\left\{1,\sqrt{\frac{dA2^{sj}}{\eps n}}\right\}
       +2^{-J}\right\} \\
    &\le C_s\sqrt d\left(\frac{dA}{\eps n}\right)^{1/s}.
\end{align*}
Combining the last display with the resolution-error bound and the triangle inequality proves \Cref{eq:fully-private-intrinsic-privtree-rate}.  The dependence of $C_s$ may be taken as in \Cref{rem:intrinsic-laplacian-s-dependence}.
\end{proof}

\section{Randomly shifted partitions}\label{sec:random-shift}

Our main theorem on the accuracy of PrivTree, \Cref{thm:fully-private-intrinsic-privtree}, provide the accuracy bound with only a weak dependence on the ambient dimension $d$. However, the factor $A$ in \Cref{ass:finite-scale-intrinsic-dimension} may be exponential in $d$ when the data is aligned with the grid.

For example, consider a cluster of data points uniformly located in a ball of radius $\eta$ centered at the center of $\Omega=[0,1]^d$, where $\eta$ can be sufficiently small. Then the data set has intrinsic dimension $s=0$ as its covering number is $O(1)$, while the first $d$ steps of the binary partitioning will bisect each coordinate once, leading to $2^d$ many nonempty cells. So $A=2^d$. In this section we show that such a factor can be avoided by randomly shifting the grid publicly before the partitioning.

Draw $U=(U_1,\ldots,U_d)$ uniformly from $[0,1]^d$ and translate each data point $x$ to $x+U$.  The translated data lie in the public window $\Omega_U\coloneqq\Omega+U\subseteq\Omega_2\coloneqq[0,2]^d$.  We use the cyclic binary partition of $\Omega_2$.  As before, the side lengths of a depth-$k$ cell remain of order $2^{-k/d}$, up to the factor $2$ from the larger domain.

During the partition, a reached cell $Q$ with $Q\cap\Omega_U=\varnothing$ is publicly known to have zero count and is discarded with neither further splitting nor privacy noise.  
For every active leaf $Q$, we choose its representative $y_Q\in Q\cap\Omega_U$, so every representative remains in $\Omega$ after translation by $-U$. 
Throughout this section, for any measure $\mu$ on $\Omega$, we write $\mu^U(E)\coloneqq\mu(E-U)$ for its translation to $\Omega_U$; the same notation without $U$ denotes the corresponding measure translated back to $\Omega$.

\begin{algorithm}[H]
\caption{Randomly Shifted PrivTree}
\label{alg:randomly-shifted-privtree}
\begin{algorithmic}[1]
\Require Empirical measure $\mu_n$ of data set $\{X_i\}_{i=1}^n\subset \Omega$ and privacy budget $\eps>0$
\State Draw $U\sim\operatorname{Unif}([0,1]^d)$ independently of the data, set $\Omega_U\gets\Omega+U$, and set $\mu_n^U\gets n^{-1}\sum_{i=1}^n\delta_{X_i+U}$.
\State Run \Cref{alg:binary-privtree} on $\mu_n^U$ over $\Omega_2$, discarding every cell disjoint from $\Omega_U$ without adding noise or further splitting, with $\theta=0$, $\sigma=6/(\eps n)$, and $\delta=\sigma\log2$.
\State Run \Cref{alg:private-synthetic-measure} with privacy budget $\eps/2$ on the active leaves, choosing each representative $y_Q$ in $Q\cap\Omega_U$, and denote the output by $\widehat\nu^U$.
\State \Return the measure $\widehat\nu$ obtained from $\widehat\nu^U$ by translating each atom $y_Q$ to $y_Q-U$.
\end{algorithmic}
\end{algorithm}

The next lemma is the geometric reason for introducing the shift.
Let $\mathcal D_k^{(2)}$ be the depth-$k$ cyclic partition of $\Omega_2$.  For a measure $\mu$ on $\Omega$, write
\[
    N_{k,U}^+(\mu)\coloneqq\#\{Q\in\mathcal D_k^{(2)}:\mu^U(Q)>0\}.
\]
For data sets in $\Omega_2$, define $N_k^{\mathrm{eff}}$ by \Cref{eq:effective-occupancy} with $\mathcal D_k^{(2)}$ in place of $\mathcal D_k$.

\begin{lemma}[Averaged occupancy of a shifted grid]\label{lem:shifted-grid-occupancy}
Let $S\subseteq\Omega$ satisfy $\mu(S)=1$.
For every integer $k\ge1$,
\begin{equation}\label{eq:shifted-grid-covering}
    \mathbb E_U\bigl[N_{k,U}^+(\mu)\bigr]
    \le e\,N\!\left(S,\frac{2^{-k/d}}{2d}\right).
\end{equation}
\end{lemma}

\begin{proof}
Denote $\rho=2^{-k/d}/(2d)$ for simplicity and cover $S$ by $N(S,\rho)$ Euclidean balls of radius $\rho$. 
For each covering ball, we consider its contribution to $N_{k,U}^+(\mu)$, which is at most the number of cells in $\mathcal D_k^{(2)}$ that intersect the ball.

Fix a covering ball, and let $\zeta_i$ indicate whether its projection onto the $i$th coordinate crosses a depth-$k$ grid boundary, with $\zeta_i=0$ if that coordinate has not been split.  Every cell in $\mathcal D_k^{(2)}$ has side length at least $2^{-k/d}$, and each nontrivial grid phase is uniform modulo its mesh size.  Hence
\[
    \E_U [\zeta_i] = \mathbb P_U(\zeta_i=1)\le 2\rho 2^{k/d}=\frac1d.
\]
The ball meets at most $\prod_{i=1}^d(1+\zeta_i)=2^{\sum_{i=1}^d\zeta_i}$ cells.  Since the $\zeta_i$ are independent, summing over the covering balls gives
\[
    \mathbb E_U\bigl[N_{k,U}^+(\mu)\bigr]
    \le N(S,\rho)\prod_{i=1}^d\bigl(1+\E_U [\zeta_i]\bigr)
    \le N(S,\rho)\left(1+\frac1d\right)^d
    \le eN(S,\rho).
\]
This proves \Cref{eq:shifted-grid-covering}.
\end{proof}

We now apply the averaged profile directly in the proof of the private intrinsic-dimensional bound.

\begin{theorem}[Randomly shifted PrivTree with intrinsic-dimensional accuracy]\label{thm:shifted-covering-privtree}
Let $0<\eps\le1$ and $\eps n>2$, and follow \Cref{alg:randomly-shifted-privtree}.  Suppose that $d\ge3$, $2<s\le d$, $1\le A<\eps n$, and the data set $X=\{X_i\}_{i=1}^n\subseteq\Omega$ satisfies
\begin{equation}\label{eq:finite-scale-euclidean-covering}
    N(X,r)\le Ar^{-s}
    \qquad\text{for every }\left(\frac{A}{\eps n}\right)^{1/s}\le r\le1.
\end{equation}
Then the expected support size of the output measure is at most $C\eps n$, and the output measure $\widehat\nu$ is $\eps$-differentially private and satisfies
\begin{equation}\label{eq:shifted-covering-privtree-rate}
    \mathbb E\bigl[W_1(\mu_n,\widehat\nu)\bigr]
    \le C_s d^{3/2}\left(\frac{Ad\log(\eps n)}{\eps n}\right)^{1/s}.
\end{equation}
The expectation is over the public shift and the noise used in the split decisions and the leaf-mass release.
\end{theorem}

\begin{proof}
Conditional on $U$, the discarded cells are determined by the public set $\Omega_U$, and the active tree has the same law as the restriction to $\Omega_U$ of an ordinary PrivTree on $\Omega_2$.  Thus privacy and the expected support bound follow as in \Cref{thm:privtree_1/d}; averaging over the data-independent shift preserves privacy.

If $eA(2d)^s\ge\eps n$, then \Cref{eq:shifted-covering-privtree-rate} follows from $W_1\le\sqrt d$.  Otherwise, for every
\[
    1\le k\le\left\lceil\frac ds\log_2\!\left(\frac{\eps n}{eA(2d)^s}\right)\right\rceil,
\]
the covering radius in \Cref{lem:shifted-grid-occupancy} lies in the range of \Cref{eq:finite-scale-euclidean-covering}, and hence
\[
    \mathbb E_U\bigl[N_{k,U}^+(\mu_n)\bigr]
    \le eA(2d)^s2^{sk/d}.
\]
The same bound at $k=0$ is immediate.
The proof of \Cref{thm:fully-private-intrinsic-privtree} now applies after averaging over $U$, with $eA(2d)^s$ in place of $A$.  Indeed, the resolution argument uses the concavity of $x\mapsto\min\{1,cx\}$, while the recurrence for the number of active nodes is linear.  Since translation is an isometry and rescaling $\Omega_2$ contributes only a universal factor,
\[
    \mathbb E\bigl[W_1(\mu_n,\widehat\nu)\bigr]
    \le C_s\sqrt d\left(\frac{eA(2d)^sd\log(\eps n)}{\eps n}\right)^{1/s}
    \le C_s d^{3/2}\left(\frac{Ad\log(\eps n)}{\eps n}\right)^{1/s}.
\]
This proves \Cref{eq:shifted-covering-privtree-rate}; the dependence of $C_s$ may be taken as in \Cref{rem:intrinsic-laplacian-s-dependence}.
\end{proof}

\begin{remark}[Comparison with the fixed grid]\label{rem:random-shift-limitation}
For the fixed partition in \Cref{thm:fully-private-intrinsic-privtree}, a ball of radius $2^{-\lceil k/d\rceil-1}$ can contribute $2^d$ many cells in $N^+_k$. Therefore, as pointed out at the beginning of this section,
\[
    N_k^+(\mu_n)
    \le 2^dN\!\left(X,2^{-\lceil k/d\rceil-1}\right)
    \le 2^dN\!\left(X,\frac{2^{-k/d}}4\right).
\]
This comparison introduces a factor $2^{d/s}$ in the Wasserstein bound of \Cref{thm:fully-private-intrinsic-privtree}. On the other hand, the random shift replaces $2^{d/s}$ by the factor $d$ in \Cref{eq:shifted-covering-privtree-rate}.  This is an exponential-to-polynomial improvement in $d$ when $s$ is small relative to $d$; when $s$ is comparable to $d$, the fixed-grid comparison can instead be sharper.
\end{remark}

\section*{Acknowledgments}
The author thanks Roman Vershynin for helpful discussions and valuable insights. 
The author also thanks OpenAI's GPT-5.6 Sol for assistance in formalizing the author's proof outlines into initial LaTeX drafts.


\appendix

\section{Capped adaptive binary partition}\label{app:capped-adaptive-partition}

This section gives a capped version of \Cref{alg:naive-partition}. It stops when the number of leaves reaches a prescribed cap, typically after $O(n)$ splits. The breadth-first rule is needed for the error bound: among the leaves whose mass exceeds $\theta$, the algorithm splits one of minimum depth. The queue in \Cref{alg:capped-partition} enforces this rule.

\begin{algorithm}[H]
\caption{Capped adaptive binary partition}\label{alg:capped-partition}
\begin{algorithmic}[1]
\Require Probability measure $\mu$ on $[0,1]^d$, a threshold $\theta\in(0,1]$, and a leaf cap $M\ge1$
\State Initialize $\mathcal T_{\theta,M}$ with its root $\Omega$.
\State Initialize a queue $\mathcal H$ with $\mathcal H\gets\{\Omega\}$.
\While{$\mathcal H\ne\varnothing$ and $|\mathcal L(\mathcal T_{\theta,M})|<M$}
    \State Pop the front cell $Q$ from $\mathcal H$.
    \If{$\mu(Q)>\theta$}
        \State Bisect $Q$ by the cyclic coordinate rule, obtaining $Q_0$ and $Q_1$, and add them to $\mathcal T_{\theta,M}$.
        \State Append $Q_0$ and $Q_1$ to the back of $\mathcal H$.
    \EndIf
\EndWhile
\State \Return $\mathcal T_{\theta,M}$.
\end{algorithmic}
\end{algorithm}

Apply \Cref{alg:exact-synthetic-measure} to $\mu$ and the output tree $\mathcal T_{\theta,M}$, and write its output as $\nu_{\mathcal T_{\theta,M}}$. Each split increases the number of leaves by one, so \Cref{alg:capped-partition} terminates after at most $M-1$ splits. The next theorem shows that the cap preserves the Wasserstein rate.

\begin{theorem}
For every probability measure $\mu$ on $[0,1]^d$, the exact synthetic measure associated with the output tree of \Cref{alg:capped-partition} satisfies
\begin{equation}\label{eq:capped-w1}
    W_1(\mu,\nu_{\mathcal T_{\theta,M}})\le 2\sqrt d\left(\theta^{1/d}+\left(\frac{2}{M}\right)^{1/d}\right).
\end{equation}
In particular, if $\theta=c/n$ and $M=\lceil Cn\rceil$, where $c,C>0$ are independent of $n$, then
\[
    W_1(\mu,\nu_{\mathcal T_{\theta,M}})\le C'\sqrt d\,n^{-1/d}.
\]
\end{theorem}

\begin{proof}
If $M=1$, the trivial bound $W_1(\mu,\nu_{\mathcal T_{\theta,M}})\le\sqrt d$ holds. Assume that $M\ge2$. Let $\mathcal L=\mathcal L(\mathcal T_{\theta,M})$ and divide the final leaves into
\[
    \mathcal L_{\mathrm{light}}=\{Q\in\mathcal L:p_Q\le\theta\},
    \qquad
    \mathcal L_{\mathrm{heavy}}=\{Q\in\mathcal L:p_Q>\theta\}.
\]
By \Cref{lem:resolution-error}, it suffices to bound the light and heavy contributions to $\mathfrak R_{\mathcal T_{\theta,M}}(\mu)$.

\begin{samepage}
The proof of \Cref{lem:leafwise-holder} also applies when both sums are restricted to any $\mathcal A\subseteq\mathcal L$. Applying it with $\mathcal A=\mathcal L_{\mathrm{light}}$, together with $p_Q\le\theta$ on light leaves and \Cref{eq:kraft}, gives
\begin{equation}\label{eq:light-leaves}
    \sum_{Q\in\mathcal L_{\mathrm{light}}}p_Q2^{-\depth(Q)/d}\le\theta^{1/d}.
\end{equation}
\end{samepage}

The heavy contribution is zero if $\mathcal L_{\mathrm{heavy}}$ is empty. Otherwise, the algorithm stopped upon reaching the cap, so $|\mathcal L|=M$. Let $h=\min_{Q\in\mathcal L_{\mathrm{heavy}}}\depth(Q)$. The depths of the split nodes are nondecreasing under the breadth-first rule, so no final leaf has depth greater than $h+1$. Thus $\mathcal T_{\theta,M}$ has depth at most $h+1$, and all heavy leaves have depth $h$ or $h+1$. Therefore, $M\le2^{h+1}$ and
\begin{equation}\label{eq:heavy-leaves}
\begin{split}
    \sum_{Q\in\mathcal L_{\mathrm{heavy}}}p_Q2^{-\depth(Q)/d}
    &\le 2^{-h/d}\sum_{Q\in\mathcal L_{\mathrm{heavy}}}p_Q
    \le 2^{-h/d}
    \le \left(\frac{2}{M}\right)^{1/d}.
\end{split}
\end{equation}
Combining \Cref{eq:light-leaves,eq:heavy-leaves} proves \Cref{eq:capped-w1}. The stated specialization follows from $M=\lceil Cn\rceil\ge Cn$.
\end{proof}

\section{Linear-programming form of the PSMM projection}
\label{app:psmm-projection}

Enumerate the leaves as $Q_1,\ldots,Q_m$ and write $\widetilde p_i=\widetilde p_{Q_i}$, $c_{ij}=\|y_{Q_i}-y_{Q_j}\|_2$, and $D=\sqrt d$.  The projection in \Cref{eq:psmm-projection} is equivalently obtained from the linear program
\begin{equation}\label{eq:psmm-projection-lp}
\begin{aligned}
\underset{p,\gamma,v^+,v^-}{\operatorname{minimize}}\quad
    &\sum_{i\ne j}c_{ij}\gamma_{ij}
      +D\sum_{i=1}^m(v_i^++v_i^-)\\
\text{subject to}\quad
    &\sum_{j\ne i}(\gamma_{ij}-\gamma_{ji})+v_i^+-v_i^-
       =\widetilde p_i-p_i, && i\in[m],\\
    &\sum_{i=1}^m p_i=1,\\
    &p_i,\gamma_{ij},v_i^+,v_i^-\ge0, && i,j\in[m],\ i\ne j.
\end{aligned}
\end{equation}
If $p^*$ is an optimal solution, then $\widehat p_{Q_i}=p_i^*$.  Indeed, the inner maximization in \Cref{eq:psmm-projection} is a linear program in $(f_{Q_i})_{i=1}^m$ with constraints $f_{Q_i}-f_{Q_j}\le c_{ij}$ and $-D\le f_{Q_i}\le D$.  Its feasible set is symmetric, so the absolute value does not change the maximum, and linear-programming duality gives \Cref{eq:psmm-projection-lp}.  This is the Euclidean-metric version of \cite[Algorithm~2 and Proposition~6]{he2023algorithmically}.

\bibliographystyle{plain}
\bibliography{ref}

@inproceedings{zhang2016privtree,
  title={{PrivTree}: A Differentially Private Algorithm for Hierarchical Decompositions},
  author={Zhang, Jun and Xiao, Xiaokui and Xie, Xing},
  booktitle={Proceedings of the 2016 International Conference on Management of Data},
  pages={155--170},
  year={2016},
  publisher={ACM},
  doi={10.1145/2882903.2882928}
}

@article{dwork2014algorithmic,
  title={The Algorithmic Foundations of Differential Privacy},
  author={Dwork, Cynthia and Roth, Aaron},
  journal={Foundations and Trends in Theoretical Computer Science},
  volume={9},
  number={3--4},
  pages={211--407},
  year={2014},
  publisher={now Publishers Inc.},
  doi={10.1561/0400000042}
}

@article{hawes2020implementing,
  title={Implementing Differential Privacy: Seven Lessons From the 2020 {United States Census}},
  author={Hawes, Michael B.},
  journal={Harvard Data Science Review},
  volume={2},
  number={2},
  year={2020},
  doi={10.1162/99608f92.353c6f99},
  url={https://doi.org/10.1162/99608f92.353c6f99}
}

@article{abowd2022topdown,
  title={The 2020 {Census} Disclosure Avoidance System {TopDown} Algorithm},
  author={Abowd, John M. and Ashmead, Robert and Cumings-Menon, Ryan and Garfinkel, Simson and Heineck, Micah and Heiss, Christine and Johns, Robert and Kifer, Daniel and Leclerc, Philip and Machanavajjhala, Ashwin and Moran, Brett and Sexton, William and Spence, Matthew and Zhuravlev, Pavel},
  journal={Harvard Data Science Review},
  year={2022},
  doi={10.1162/99608f92.529e3cb9},
  url={https://doi.org/10.1162/99608f92.529e3cb9},
  note={Special Issue 2: Differential Privacy for the 2020 U.S. Census}
}

@article{hauer2021census,
  title={Differential Privacy in the 2020 {Census} Will Distort {COVID}-19 Rates},
  author={Hauer, Mathew E. and Santos-Lozada, Alexis R.},
  journal={Socius: Sociological Research for a Dynamic World},
  volume={7},
  pages={2378023121994014},
  year={2021},
  doi={10.1177/2378023121994014}
}

@inproceedings{he2023algorithmically,
  title={Algorithmically Effective Differentially Private Synthetic Data},
  author={He, Yiyun and Vershynin, Roman and Zhu, Yizhe},
  booktitle={Proceedings of Thirty Sixth Conference on Learning Theory},
  series={Proceedings of Machine Learning Research},
  volume={195},
  pages={3941--3968},
  year={2023},
  publisher={PMLR},
  url={https://proceedings.mlr.press/v195/he23a.html}
}

@article{boedihardjo2024private,
  title={Private Measures, Random Walks, and Synthetic Data},
  author={Boedihardjo, March and Strohmer, Thomas and Vershynin, Roman},
  journal={Probability Theory and Related Fields},
  volume={189},
  number={1--2},
  pages={569--611},
  year={2024},
  publisher={Springer},
  doi={10.1007/s00440-024-01279-z},
  url={https://doi.org/10.1007/s00440-024-01279-z}
}

@article{boedihardjo2024covariance,
  title={Covariance's Loss Is Privacy's Gain: Computationally Efficient, Private and Accurate Synthetic Data},
  author={Boedihardjo, March and Strohmer, Thomas and Vershynin, Roman},
  journal={Foundations of Computational Mathematics},
  volume={24},
  number={1},
  pages={179--226},
  year={2024},
  publisher={Springer},
  doi={10.1007/s10208-022-09591-7},
  url={https://doi.org/10.1007/s10208-022-09591-7}
}

@article{weed2019sharp,
  title={Sharp asymptotic and finite-sample rates of convergence of empirical measures in {Wasserstein} distance},
  author={Weed, Jonathan and Bach, Francis},
  journal={Bernoulli},
  volume={25},
  number={4A},
  pages={2620--2648},
  year={2019},
  doi={10.3150/18-BEJ1065}
}

@article{ba2011sublinear,
  title={Sublinear Time Algorithms for Earth Mover's Distance},
  author={{Do Ba}, Khanh and Nguyen, Huy L. and Nguyen, Huy N. and Rubinfeld, Ronitt},
  journal={Theory of Computing Systems},
  volume={48},
  number={2},
  pages={428--442},
  year={2011},
  doi={10.1007/s00224-010-9265-8}
}

@article{dereich2013constructive,
  title={Constructive Quantization: Approximation by Empirical Measures},
  author={Dereich, Steffen and Scheutzow, Michael and Schottstedt, Reik},
  journal={Annales de l'Institut Henri Poincar\'e, Probabilit\'es et Statistiques},
  volume={49},
  number={4},
  pages={1183--1203},
  year={2013},
  doi={10.1214/12-AIHP489}
}

@article{vonluxburg2004distance,
  title={Distance-Based Classification with {Lipschitz} Functions},
  author={von Luxburg, Ulrike and Bousquet, Olivier},
  journal={Journal of Machine Learning Research},
  volume={5},
  pages={669--695},
  year={2004}
}

@book{villani2009optimal,
  title={Optimal Transport: Old and New},
  author={Villani, C{\'e}dric},
  series={Grundlehren der Mathematischen Wissenschaften},
  volume={338},
  publisher={Springer},
  address={Berlin, Heidelberg},
  year={2009},
  doi={10.1007/978-3-540-71050-9}
}

@article{blum2013learning,
  title={A Learning Theory Approach to Noninteractive Database Privacy},
  author={Blum, Avrim and Ligett, Katrina and Roth, Aaron},
  journal={Journal of the ACM},
  volume={60},
  number={2},
  pages={12:1--12:25},
  year={2013},
  doi={10.1145/2450142.2450148}
}

@inproceedings{hardt2012simple,
  title={A Simple and Practical Algorithm for Differentially Private Data Release},
  author={Hardt, Moritz and Ligett, Katrina and McSherry, Frank},
  booktitle={Advances in Neural Information Processing Systems},
  volume={25},
  pages={2348--2356},
  publisher={Curran Associates, Inc.},
  year={2012},
  url={https://proceedings.neurips.cc/paper/2012/hash/208e43f0e45c4c78cafadb83d2888cb6-Abstract.html}
}

@article{ullman2020pcps,
  title={{PCP}s and the Hardness of Generating Synthetic Data},
  author={Ullman, Jonathan and Vadhan, Salil},
  journal={Journal of Cryptology},
  volume={33},
  number={4},
  pages={2078--2112},
  year={2020},
  doi={10.1007/s00145-020-09363-y}
}

@inproceedings{donhauser2024certified,
  title={Certified Private Data Release for Sparse {L}ipschitz Functions},
  author={Donhauser, Konstantin and Lokna, Johan and Sanyal, Amartya and Boedihardjo, March and H\"onig, Robert and Yang, Fanny},
  booktitle={Proceedings of The 27th International Conference on Artificial Intelligence and Statistics},
  series={Proceedings of Machine Learning Research},
  volume={238},
  pages={1396--1404},
  year={2024},
  publisher={PMLR},
  url={https://proceedings.mlr.press/v238/donhauser24a.html}
}

@article{he2025lowdimensional,
  title={Differentially Private Low-Dimensional Synthetic Data from High-Dimensional Datasets},
  author={He, Yiyun and Strohmer, Thomas and Vershynin, Roman and Zhu, Yizhe},
  journal={Information and Inference: A Journal of the IMA},
  volume={14},
  number={1},
  pages={iaae034},
  year={2025},
  doi={10.1093/imaiai/iaae034}
}

@article{fournier2015rate,
  title={On the Rate of Convergence in {Wasserstein} Distance of the Empirical Measure},
  author={Fournier, Nicolas and Guillin, Arnaud},
  journal={Probability Theory and Related Fields},
  volume={162},
  number={3--4},
  pages={707--738},
  year={2015},
  doi={10.1007/s00440-014-0583-7}
}

@inproceedings{kreacic2023kdtrees,
  title={Differentially Private Synthetic Data Using {KD}-Trees},
  author={Krea\v{c}i\'{c}, Eleonora and Nouri, Navid and Potluru, Vamsi K. and Balch, Tucker and Veloso, Manuela},
  booktitle={Proceedings of the Thirty-Ninth Conference on Uncertainty in Artificial Intelligence},
  series={Proceedings of Machine Learning Research},
  volume={216},
  pages={1143--1153},
  year={2023},
  publisher={PMLR},
  url={https://proceedings.mlr.press/v216/kreacic23a.html}
}

@inproceedings{gonzalezlara2025private,
  title={Private Evolution Converges},
  author={Gonz{\'a}lez Lara, Tom{\'a}s and Fanti, Giulia and Ramdas, Aaditya},
  booktitle={Advances in Neural Information Processing Systems},
  volume={38},
  pages={164688--164724},
  year={2025},
  publisher={Neural Information Processing Systems Foundation, Inc.},
  doi={10.52202/085713-4961},
  url={https://proceedings.neurips.cc/paper_files/paper/2025/hash/da339ea9d0bcb0fc1a5ab07961f022bb-Abstract-Conference.html}
}

@inproceedings{dwork2006calibrating,
  title={Calibrating Noise to Sensitivity in Private Data Analysis},
  author={Dwork, Cynthia and McSherry, Frank and Nissim, Kobbi and Smith, Adam},
  booktitle={Theory of Cryptography},
  series={Lecture Notes in Computer Science},
  volume={3876},
  pages={265--284},
  publisher={Springer},
  year={2006},
  doi={10.1007/11681878_14}
}

@inproceedings{liu2021robustmean,
  title={Robust and Differentially Private Mean Estimation},
  author={Liu, Xiyang and Kong, Weihao and Kakade, Sham M. and Oh, Sewoong},
  booktitle={Advances in Neural Information Processing Systems},
  volume={34},
  pages={3887--3901},
  year={2021},
  publisher={Curran Associates, Inc.},
  url={https://proceedings.neurips.cc/paper/2021/hash/1fc5309ccc651bf6b5d22470f67561ea-Abstract.html}
}

@inproceedings{amin2019covariance,
  title={Differentially Private Covariance Estimation},
  author={Amin, Kareem and Dick, Travis and Kulesza, Alex and Mu\~noz Medina, Andr\'es and Vassilvitskii, Sergei},
  booktitle={Advances in Neural Information Processing Systems},
  volume={32},
  year={2019},
  publisher={Curran Associates, Inc.},
  url={https://proceedings.neurips.cc/paper/2019/hash/4158f6d19559955bae372bb00f6204e4-Abstract.html}
}

@inproceedings{abadi2016deep,
  title={Deep Learning with Differential Privacy},
  author={Abadi, Martin and Chu, Andy and Goodfellow, Ian and McMahan, H. Brendan and Mironov, Ilya and Talwar, Kunal and Zhang, Li},
  booktitle={Proceedings of the 2016 ACM SIGSAC Conference on Computer and Communications Security},
  pages={308--318},
  year={2016},
  publisher={ACM},
  doi={10.1145/2976749.2978318}
}

@inproceedings{sheffet2017ols,
  title={Differentially Private Ordinary Least Squares},
  author={Sheffet, Or},
  booktitle={Proceedings of the 34th International Conference on Machine Learning},
  series={Proceedings of Machine Learning Research},
  volume={70},
  pages={3105--3114},
  publisher={PMLR},
  year={2017},
  url={https://proceedings.mlr.press/v70/sheffet17a.html}
}

@inproceedings{ghazi2020clustering,
  title={Differentially Private Clustering: Tight Approximation Ratios},
  author={Ghazi, Badih and Kumar, Ravi and Manurangsi, Pasin},
  booktitle={Advances in Neural Information Processing Systems},
  volume={33},
  pages={4040--4054},
  year={2020},
  publisher={Curran Associates, Inc.},
  url={https://proceedings.neurips.cc/paper/2020/hash/299dc35e747eb77177d9cea10a802da2-Abstract.html}
}

@inproceedings{lin2024foundation,
  title={Differentially Private Synthetic Data via Foundation Model {API}s 1: Images},
  author={Lin, Zinan and Gopi, Sivakanth and Kulkarni, Janardhan and Nori, Harsha and Yekhanin, Sergey},
  booktitle={International Conference on Learning Representations},
  year={2024},
  url={https://openreview.net/forum?id=YEhQs8POIo},
  eprint={2305.15560},
  archivePrefix={arXiv}
}

@inproceedings{ding2026smooth,
  title={Minimax Optimal Differentially Private Synthetic Data for Smooth Queries},
  author={Ding, Rundong and He, Yiyun and Zhu, Yizhe},
  booktitle={Proceedings of Thirty Ninth Conference on Learning Theory},
  series={Proceedings of Machine Learning Research},
  volume={336},
  pages={1963--1964},
  publisher={PMLR},
  year={2026},
  url={https://proceedings.mlr.press/v336/ding26a.html}
}

@inproceedings{feldman2024instance,
  title={Instance-Optimal Private Density Estimation in the {Wasserstein} Distance},
  author={Feldman, Vitaly and McMillan, Audra and Sivakumar, Satchit and Talwar, Kunal},
  booktitle={Advances in Neural Information Processing Systems},
  volume={37},
  pages={90061--90131},
  year={2024},
  publisher={Neural Information Processing Systems Foundation, Inc.},
  doi={10.52202/079017-2860}
}

@article{ghazi2023heatmaps,
  title={Differentially Private Heatmaps},
  author={Ghazi, Badih and He, Junfeng and Kohlhoff, Kai and Kumar, Ravi and Manurangsi, Pasin and Navalpakkam, Vidhya and Valliappan, Nachiappan},
  journal={Proceedings of the AAAI Conference on Artificial Intelligence},
  volume={37},
  number={6},
  pages={7696--7704},
  year={2023},
  doi={10.1609/aaai.v37i6.25933}
}

\end{document}